\documentclass[11pt, onecolumn]{IEEEtran}
\usepackage{mathrsfs}
\usepackage{amsfonts}
\usepackage{amsfonts}
\usepackage{amsfonts}
\usepackage{amsfonts}
\usepackage{amsfonts}
\usepackage{amsfonts}
\usepackage{amsfonts}
\usepackage{amsfonts}
\usepackage{amsfonts}
\usepackage{amsfonts}
\usepackage{amsfonts}
\usepackage{amsfonts}
\usepackage{amsfonts}
\usepackage{amsfonts}
\usepackage{amsfonts}
\usepackage{amsfonts}
\usepackage{bbm}
\usepackage{amsfonts}
\usepackage{amssymb}
\usepackage{stfloats}
\usepackage{cite}
\usepackage{graphicx}
\usepackage{psfrag}
\usepackage{subfigure}
\usepackage{amsmath}
\usepackage{array}
\usepackage{multirow}
\usepackage{tabularx}
\usepackage{bm}
\usepackage{color}
\usepackage{algorithm}

\begin{document}

\title{Distributive Power Control Algorithm for Multicarrier Interference Network over Time-Varying Fading Channels -- Tracking Performance Analysis and Optimization}

\newtheorem{Thm}{Theorem}
\newtheorem{Lem}{Lemma}
\newtheorem{Cor}{Corollary}
\newtheorem{Def}{Definition}
\newtheorem{Exam}{Example}
\newtheorem{Alg}{Algorithm}
\newtheorem{Prob}{Problem}
\newtheorem{Rem}{Remark}
\newtheorem{Proof}{Proof}
\newtheorem{Obs}{Observation}
\newtheorem{Con}{Conclusion}
\newtheorem{Ass}{Assumption}

\newcommand{\argmin}[1]{\underset{#1}{\operatorname{argmin}}\;}
\newcommand{\argmax}[1]{\underset{#1}{\operatorname{argmax}}\;}
\newcommand{\maximize}[1]{\underset{#1}{\operatorname{maximize}}\;}
\newcommand{\minimize}[1]{\underset{#1}{\operatorname{minimize}}\;}

\author{\IEEEauthorblockN{Yong Cheng, Vincent K. N. Lau\\}
\authorblockA{Department of Electronic and Computer Engineering \\
The Hong Kong University of Science and Technology}}
\maketitle

\begin{abstract}
Distributed power control over interference limited network has received an increasing intensity of interest over the past few years. Distributed solutions (like the iterative water-filling, gradient projection, etc.) have been intensively investigated under \emph{quasi-static} channels. However, as such distributed solutions involve iterative updating and explicit message passing, it is unrealistic to assume that the wireless channel remains unchanged during the iterations. Unfortunately, the behavior of those distributed solutions under \emph{time-varying} channels is in general unknown.  In this paper, we shall investigate the distributed scaled gradient projection algorithm (DSGPA) in a $K$ pairs multicarrier interference network under a finite-state Markov channel (FSMC) model.  We shall analyze the \emph{convergence property} as well as \emph{tracking performance} of the proposed DSGPA. Our analysis shows that the proposed DSGPA converges to a limit region rather than a single point under the  FSMC model. We also show that the order of growth of the tracking errors is given by $\mathcal{O}\left(1 \big/ \overline{N}\right)$, where $\overline{N}$ is the \emph{average sojourn time} of the FSMC. Based on the analysis, we shall derive the \emph{tracking error optimal scaling matrices} via Markov decision process modeling. We shall show that the tracking error optimal scaling matrices can be implemented distributively at each transmitter. The numerical results show the superior performance of the proposed DSGPA over three baseline schemes, such as the gradient projection algorithm with a constant stepsize.
\end{abstract}


\IEEEpeerreviewmaketitle

\section{Introduction}\label{sec:intro}
Power control algorithm design over interference limited network has received an increasing intensity of interest over the past few years. In \cite{Yu2002}, the power control design over $K$ pairs interference network is formulated as a deterministic non-cooperative game and distributive solution, namely the {\em iterative water-filling} algorithm is proposed to achieve the Nash Equilibrium (NE) of the game. The distributed power control algorithm design in interference network has also been studied in \cite{   Scutari2008_I, Scutari2008II } using game theory, and from the distributed network utility maximization (NUM) point of view \cite{Palomar2006, Chiang2005}. \textcolor{black}{There are also some other works on distributed power allocation in interference networks, such as the asynchronous distributed pricing
(ADP) framework \cite{Huang2005, Huang2006} and the adjoint network\footnote{\textcolor{black}{In the adjoint network approach, the authors \cite{Stanczak2007, Stanczak2008} first derive a computational algorithm to obtain a global optimal solution. The computational algorithm requires global observations across the wireless networks. The authors \cite{Stanczak2007, Stanczak2008} then proposed an efficient mechanism based on adjoint network to distribute the global observations across the nodes.}}  based approach \cite{Stanczak2007, Stanczak2008}.} In all these works, distributive solutions are critical in networks where many transmitter-receiver pairs are randomly placed. In the absence of an infrastructure linking all these nodes, centralized solutions are difficult to implement due to the difficulty of gathering global knowledge of channel state information (CSI) as well as non-scalable limitations in terms of complexity. The implementation of such distributive algorithms, such as the distributed gradient projection algorithm \cite{Scutari2008II}, the iterative water-filling algorithm \cite{Yu2002}, and the gossip algorithm \cite{Boyd2006}, often involves iterative solution with explicit message passing. Moreover, in all these existing works, the convergence and the optimality properties of the algorithms are established under the quasi-static channel assumption. Specifically, the CSI is assumed to be quasi-static throughout the iteration process. However, since there are explicit message passing in between each iterative step, it is quite unrealistic to assume that the channel remains unchanged over a significant number of iterative steps for the algorithm to converge. As a result, it is of great importance both theoretically and practically to investigate the behavior of the distributive algorithms \cite{Yu2002, Scutari2008_I , Palomar2006, Chiang2005, Scutari2008II, Boyd2006, Huang2005, Huang2006, Stanczak2007, Stanczak2008} under time-varying channels. In this paper, we are interested in a network topology consisting of $K$ transmit-receive pairs sharing a common spectrum with $N_{F}$ independent subbands. For a distributive power control algorithm in such interference networks over time-varying channels, the NE, which is a function of the CSI, will be time-varying as well. Hence, the following are some important questions, which still require more investigations.
\begin{itemize}
\item \textcolor{black}{How fast could the wireless channel change before the distributive algorithm failed to track the {\em moving} NE?}
\item Given that the distributed algorithms can track the moving NE of the interference network, can we obtain closed-form bounds on the {\em tracking errors}?
\item How to enhance the existing distributive algorithm (designed for quasi-static CSI) to optimize the associated \emph{tracking performance}?
 \end{itemize}

Due to the randomness of the wireless channel and the nonlinear dynamics of the iteration process encountered, it is \emph{highly nontrivial} to answer the above questions in general.
There are some preliminary works on distributed power control in time-varying interference networks. For instance, in \cite{Alpcan2004}, the \emph{hybrid system} model was used to study the multicell CDMA interference game. While the existing works \cite{Alpcan2004} provide some preliminary results on the behavior of the distributed power control algorithm under time-varying channels, the techniques cannot be utilized in the \emph{vector} interference game that we are considering. For instance, the authors of \cite{Alpcan2004} established the region stability\footnote{A nonlinear dynamic system is said to be region stable if the trajectory of the system converges globally asymptotically to a limit region \cite{Podelski_Proof}.} of the {\em gradient play} of a multicell CDMA interference game. However, their results cannot be utilized in our case because they did not consider the transmit power constraint. Furthermore, to our best knowledge, none of the existing works have investigated the closed-form tracking error expressions as well as enhancing the distributive algorithm to optimize the associated tracking performance. On the other hand, parameter tracking based on the  linear regression model and least mean square (LMS) algorithm in nonstationary environments was investigated in \cite{Costa2007, Yin2005}. Also, the authors of \cite{Berenguer2005, Krishnamurthy2005}  studied the problem of tracking the optima of discrete stochastic optimization  in time-varying scenarios. However, the techniques and results of \cite{Costa2007, Yin2005, Berenguer2005, Krishnamurthy2005}, which are based on the special structure of the underlying dynamics, cannot be applied to the \emph{distributed gradient projection algorithm} we are considering.

In this paper, we shall shed some lights on the above open questions. We shall model the transient of the distributive power control algorithm in the $K$-pair interference network as an {\em algorithm trajectory} of a nonlinear system \cite{Liberzon2003}. Based on {\em randomly switched system} modeling, we first establish the \emph{region stability} and the technical conditions for the convergence of the distributive power control algorithm under a finite-state Markov channel (FSMC) model. Based on that, we shall derive closed-form order of growth of the {\em tracking errors}, namely the {\em expected-absolute-error} (EAE) as well as the mean-square-error (MSE), between the algorithm trajectory and the moving NE.  Based on these results, we shall enhance the tracking performance of the distributive power control algorithm using a novel {\em Scaling Matrix Optimization}. Specifically, we shall determine the closed-form optimal {\em scaling matrices} in the iterative update of the power control algorithm so as to minimize a general function  of the tracking errors. The optimal scaling matrix is adaptive to the current CSI and can be computed distributively at each transmitter based on local CSI only. As a result, the solution could be implemented with low complexity.

The paper is organized as follows. In Section \ref{sec:sys_mod}, we introduce the interference network model, FSMC model as well as the
formulation of the power control game. In Section \ref{sec:region_stability}, we elaborate the  distributed scaled gradient projection algorithm and analyze its convergence behavior as well as tracking errors via switched system modeling. In Section \ref{sec:mse_opt}, we shall first introduce a dominating error process and we shall then formulate the tracking error minimization problem via MDP modeling and derive the optimal solution of the scaling matrices. Section \ref{sec:num_res_disc} demonstrates the tracking performance of the
proposed algorithm and verifies the analytical results by simulations. Finally, we conclude with a summary of the main results in Section \ref{sec:conclusion}.

{\em Notations}: Matrix and vectors are denoted with capitalized and small boldface letters, respectively. $\mathbf{A}^{T}$ and $\mathbf{A}^{*}$ denote the transpose and complex conjugate of matrix $\mathbf{A}$, respectively. $\lambda_{max}\left(\mathbf{A}\right)$ and $\lambda_{min}\left(\mathbf{A}\right)$ denote the largest eigenvalue and smallest eigenvalue of matrix $\mathbf{A}$, respectively. $\left[\mathbf{A}\right]_{lm}$ denotes the $(l, m)^{th}$ entry of matrix $\mathbf{A}$, and $\mathbf{I}_{N_{F}}$ denotes the $N_{F} \times N_{F}$ identity matrix.  $\mathbb{C}$, $\mathbb{R}$ and $\mathbb{R}_{+}$ denote the set of complex numbers, real numbers and non-negative real numbers, respectively.  $\mathbb{E}$ denotes the operation of taking expectation, $\bigotimes$ denotes the operation of Cartesian product, and $Pr\left\{ \psi\right\}$ denotes the probability of set  $\psi$. Finally, $\mathbf{1}$ denotes a column vector of all ones with appropriate dimension.

\section{System Model}\label{sec:sys_mod}
In this section, we shall introduce the time-varying $K$-pair multicarrier interference network model as well as the formulation of the deterministic non-cooperative power control game.

\subsection{$K$-pair Multicarrier Interference Network Model}\label{subsec:ch_mod}
Consider a time-varying multicarrier interference network with $K$ transmitter-receiver pairs sharing $N_{F}$ non-overlapping subcarriers, as shown in Fig. \ref{fig:system_model}. Denoting $\mathcal{K}$ and $\mathcal{S}$ as the set of transmitters and subcarriers respectively, the baseband signal model at the $n^{th}$ time-slot can be written as
\begin{IEEEeqnarray}{rCl} \label{eqn:signal_model}
y_{k}^{(s)}(n) = \sum_{j =
1}^{K}h_{kj}^{(s)}(n)\sqrt{p_{j}^{(s)}(n)}d_{j}^{(s)}(n) + z_{k}^{(s)}(n),
\forall~ s \in \mathcal{S}, \forall~ k \in \mathcal{K},
\end{IEEEeqnarray}
where
\begin{itemize}
\item[--] $y_{k}^{(s)}(n)$  denotes the received signal at the $k^{th}$ receiver on  the $s^{th}$ subcarrier;
\item[--] $h_{kj}^{(s)}(n)$ denotes the channel coefficient between the $k^{th}$
receiver and the $j^{th}$ transmitter on  the $s^{th}$ subcarrier;
\item[--] $p_{j}^{(s)}(n)$ denotes the transmit power of the $j^{th}$ transmitter on  the $s^{th}$ subcarrier;
\item[--] \textcolor{black}{$d_{j}^{(s)}(n)$} denotes the transmitted
symbol of the $j^{th}$ transmitter on  the $s^{th}$ subcarrier, with normalized power, i.e.,
$\mathbb{E}\left \{d_{j}^{(s)}(n)d_{j}^{(s)*}(n)\right\}=1$;
\item[--] $z_{k}^{(s)}(n)$ denotes the additive white Gaussian
noise at the  $k^{th}$ receiver on the $s^{th}$ subcarrier, the power of which is given by
$\mathbb{E}\left\{z_{k}^{(s)}(n)z_{k}^{(s)*}(n)
\right\}=\sigma^{2}$.
\end{itemize}

Based on the signal model given in equation \eqref{eqn:signal_model}, the instantaneous mutual information (in nats per channel use) of the $k^{th}$ link at time-slot $n$ can be written as
\begin{IEEEeqnarray}{rCl} \label{eqn:Cin}
C_{k}(n) = \sum_{s = 1}^{N_{F}}\log\left(1 + \gamma_{k}^{(s)}(n)\right), \forall~ k \in \mathcal{K},
\end{IEEEeqnarray}
where $\gamma_{k}^{(s)}(n)$ denotes the receiving signal to interference plus noise ratio (SINR) at the $k^{th}$ receiver on the $s^{th}$ subcarrier, which is given by
\begin{IEEEeqnarray}{rCl} \label{eqn:sinr_model}
\gamma_{k}^{(s)}(n) = \frac{g_{kk}^{(s)}(n)p_{k}^{(s)}(n)}{\sigma^{2} +
\sum_{j = 1, j\neq k}^{K} g_{kj}^{(s)}(n)p_{j}^{(s)}(n)  }, \forall~ k \in \mathcal{K}, \forall~  s \in \mathcal{S},
\end{IEEEeqnarray}
where $g_{kj}^{(s)}(n) = \left|h_{kj}^{(s)}(n)\right|^2$ is the power gain of the fading channel  coefficient $h_{kj}^{(s)}(n)$.

\subsection{Finite-State Markov Channel Model} \label{subsec:fsmc}
Motivated by the accuracy and simplicity of the FSMC model \cite{Huang2009, Zhang2009, YIN2008, Babich2000, Zhang2000, Wang1995} for time-varying channels\footnote{\textcolor{black}{The current model of time-varying finite state Markov Fading Channel (FSMC) is a very commonly accepted model which has been widely used in a lot of literature such as \cite{Huang2009, Zhang2009, YIN2008, Babich2000, Zhang2000, Wang1995} to model the time-varying fading channels. This model is not too complicated so that it is analytically tractable and it is complicated enough to give us some first order insights. The results in this paper can be extended to the case with continuous state space but the extension to continuous state space Markov fading channel involves some  mathematical technicality. The extension  can be considered as the limit of a sequence of finite state Markov Chain (FSMC) models \cite{Diaconis1997, Wilkinson2006}. }}, we model the channel process $\left\{h_{kj}^{(s)}(n)\right\}$ as an ergodic finite-state Markov chain, $\forall~ k, j \in \mathcal{K}, s \in \mathcal{S}$.  Let the state space of the FSMC $\left\{h_{kj}^{(s)}(n) \right\}$  be $\mathcal{\widetilde{H}}_{kj}^{(s)}$, with cardinality $\widetilde{Q}_{kj}^{(s)}$, $\forall~ k, j \in \mathcal{K}, s \in \mathcal{S}$. For the FSMC $\left\{h_{kj}^{(s)}(n) \right\}$, we make the following assumptions.
\begin{Ass} [Assumptions on the FSMC] Similar to \cite{Zhang2000, Wang1995},  the transition probability matrix $\mathbf{T}_{kj}^{(s)} \in \textcolor{black}{\mathbb{R}_{+}^{ \widetilde{Q}_{kj}^{(s)} \times \widetilde{Q}_{kj}^{(s)} }}$ of the FSMC $\left\{h_{kj}^{(s)}(n)\right\}$ is assumed to have the following structure: \label{ass:assumptions_fsmc}
\begin{IEEEeqnarray}{rCl} \label{eqn:TPM}
\begin{array}{ccc}
\mathbf{T}_{kj}^{(s)} & =  &  \left[ \begin{array}{cccccccc} \nu & \varepsilon & 0 &  0 &  0 & \cdots & 0 & \varepsilon\\
\varepsilon & \nu & \varepsilon &  0 & 0 & \cdots & 0  & 0 \\
0 & \varepsilon & \nu &  \varepsilon &  0 & \cdots & 0 & 0 \\
\vdots & \vdots & \vdots & \vdots & \vdots & \ddots & \vdots & 0\\
\varepsilon & 0 & 0 &  0 &  0 & \cdots &  \varepsilon & \nu
\end{array} \right],
\end{array}
\end{IEEEeqnarray}
where $\nu = 1 - 2\varepsilon$ and $\varepsilon = \mathcal{O}\left(f_{d}\tau\right)$, with $f_{d}$ and $\tau$ denoting the doppler frequency shift and symbol duration, respectively.  ~\hfill\IEEEQEDclosed
\end{Ass}

Let $\mathbf{h}(n) \in \mathbb{C}^{K^{2}N_{F} \times 1}$ denote the collection of all the fading channel coefficients\footnote{\textcolor{black}{Here, we assume independent sub-channels for simplicity. However, all the results stated in this paper hold for correlated sub-channels as well.}}, i.e.,
\begin{IEEEeqnarray}{lCl} \label{eqn:hn}
\mathbf{h}(n) = \left\{h_{kj}^{(s)}(n),~ \forall~ (k,j,s) \in \mathcal{K} \otimes \mathcal{K} \otimes \mathcal{S}\right\} \in \mathbb{C}^{K^{2}N_{F} \times 1},
\end{IEEEeqnarray}
then $\left\{\mathbf{h}(n)\right\}$ is also an ergodic finite-state Markov chain \cite{YIN2008}. The \emph{state space} $\mathcal{H}$ and \emph{transition probability matrix} $\mathbf{T}$ of the FSMC $\left\{\mathbf{h}(n)\right\}$ are given by \cite{YIN2008}
\begin{IEEEeqnarray}{lCl} \label{eqn:hn}
\mathcal{H} = \bigotimes_{\{(k, j, s) \in \mathcal{K} \otimes \mathcal{K} \otimes \mathcal{S} \}} \mathcal{\widetilde{H}}_{kj}^{(s)},  \mbox{ and }~
\mathbf{T} = \bigotimes_{\{(k, j, s) \in \mathcal{K} \otimes \mathcal{K} \otimes \mathcal{S} \}} \mathbf{T}_{kj}^{(s)} \in \mathbb{R}_{+}^{Q \times Q},
\end{IEEEeqnarray}
where  $Q = \left| \mathcal{H}\right| = \prod_{\{(k, j, s) \in \mathcal{K} \otimes \mathcal{K} \otimes \mathcal{S} \}} \widetilde{Q}_{kj}^{(s)}$  is the cardinality of the state space $\mathcal{H} = \big\{ \mathbf{h}_{1}, \mathbf{h}_{2}, \cdots,$ $\mathbf{h}_{Q} \big\}$ of the FMSC $\left\{\mathbf{h}(n)\right\}$. \textcolor{black}{For notational convenience, we use $q \in \mathcal{Q} = \{1, 2, \cdots, Q\}$ as an index to enumerate the state realization of the CSI $\left\{\mathbf{h}(n) \in \mathcal{H}\right\}$.}

\subsection{Power Control Game in the $K$-pair Interference Network} \label{subsec:game_theoretic}
In the strategic noncooperative game formulation of the distributed power control problem in quasi-static interference networks, the players are
the $K$ active links and the payoff functions are the instantaneous mutual information  of the active links \cite{Scutari2008_I, Scutari2008II}.
As a result, for each realization of the CSI $\{\mathbf{h}(n)\}$, the power control game formulation\footnote{\textcolor{black}{The non-cooperative game considered in this paper serves more like a motivating example to study the issue of convergence of iterative algorithms under time-varying CSI.  The approach considered in this paper can also be applied to other contraction-based iterative algorithms as well\cite{Cheng2010}. }} is summarized below.
\begin{Prob} [Multicarrier Interference Game] At each time-slot $n$, given the transmit power profile $\mathbf{p}_{-k}(n)$
of the other players, the $k^{th}$ player tries to maximizes its own payoff function via solving the following capacity maximization problem: \label{prob:game}
\begin{IEEEeqnarray}{rCl} \label{eqn:game_formulation}
\begin{array}{cccc}
\left(\mathcal{G}\right): & \begin{array}{cc}
 \maximize{\mathbf{p}_{k}(n)} \\ \mbox{subject to}  \end{array}  & \begin{array}{ll} C_{k}(\mathbf{p}_{k}(n), \mathbf{p}_{-k}(n))\\
 \mathbf{p}_{k}(n) \in \Omega_{k}
\end{array}, &  \forall~ k \in \mathcal{K},
 \end{array}
\end{IEEEeqnarray}
\end{Prob}
where
\begin{itemize}
\item[--] $C_{k}(\mathbf{p}_{k}(n), \mathbf{p}_{-k}(n)) = C_{k}(n)$ is
the instantaneous capacity of the $k^{th}$ link given in \eqref{eqn:Cin};
\item[--] $\mathbf{p}_{k}(n)$ denotes the transmit power profile of the $k^{th}$ player, i.e.,
\begin{IEEEeqnarray}{rCl} \label{eqn:p_k}
 \mathbf{p}_{k}(n) = \left[p_{k}^{(1)}(n) \:\ p_{k}^{(2)}(n) \:\ \cdots \:\ p_{k}^{(N_{F})}(n) \right]^{T}, \forall~ k \in \mathcal{K};
 \end{IEEEeqnarray}
\item[--] $\mathbf{p}_{-k}(n)$ denotes the transmit power profile of all players excluding the $k^{th}$ player, i.e.,
\begin{IEEEeqnarray}{rCl} \label{eqn:p_exi}
\mathbf{p}_{-k}(n) = \left[ \mathbf{p}_{1}^{T}(n) \:\
\mathbf{p}_{2}^{T}(n) \:\ \cdots \:\ \mathbf{p}_{k-1}^{T}(n)
\:\ \mathbf{p}_{k+1}^{T}(n) \cdots \:\ \mathbf{p}_{K}^{T}(n)
\right]^{T}, \forall~ k \in \mathcal{K};
\end{IEEEeqnarray}
\item[--] the strategy set $\Omega_{k}$ of the $k^{th}$ player is given by
\begin{IEEEeqnarray}{rCl} \label{eqn:omega_k}
\Omega_{k} = \left\{\mathbf{p}_{k} \bigg | \mathbf{p}_{k} \in \mathbb{R}_{+}^{N_{F}}, \mathbf{1}^{T}\mathbf{p}_{k} \leq P_{k,max}, ~\mathbf{p}_{k} \succeq 0  \right\}, \forall~ k \in \mathcal{K},
\end{IEEEeqnarray}
with  $P_{k, max}$ denoting the maximum power budget at the $k^{th}$ transmitter.
\end{itemize}

One widely adopted optimality criterion for the game formulation given in \eqref{eqn:game_formulation} is the achievement of a Nash Equilibrium (NE) \cite{   Scutari2008_I}, which is formally  defined below.
\begin{Def} A set of (pure) strategies is a Nash Equilibrium of {\em Game $\mathcal{G}$} if no player can benefit by unilaterally changing its strategy. Mathematically, a (pure) strategy profile $\mathbf{p}^{*}(n) = \left(\mathbf{p}_{k}^{*}(n), \mathbf{p}_{-k}^{*}(n)\right)$ is a NE of  {\em Game $\mathcal{G}$} if
\begin{IEEEeqnarray}{rCl} \label{eqn:NE}
C_{k}\left(\mathbf{p}_{k}^{*}(n), \mathbf{p}_{-k}^{*}(n)\right) \geq C_{k}\left(\mathbf{p}_{k}(n), \mathbf{p}_{-k}^{*}(n)\right), ~\forall~ \mathbf{p}_{k}(n) \in \Omega_{k}, ~\forall~ k \in \mathcal{K}.
\end{IEEEeqnarray}
\end{Def}

In quasi-static interference networks, the iterative water-filling algorithm (IWFA) as well as the gradient projection algorithm (GPA) were proposed to solve {\em Game $\mathcal{G}$} \cite{Yu2002,    Scutari2008_I , Scutari2008II}. The convergence of these algorithms to the \emph{static} NE were established based on the \emph{contraction mapping} theory \cite{Bertsekas1989}, i.e., the IWFA and GPA were shown to be contraction mapping \cite{Yu2002, Scutari2008_I , Scutari2008II}. Moreover, for each update of the power allocation vector $\mathbf{p}_{k}(n)$, the  $k^{th}$ transmitter requires the feedback of local interference level at the $k^{th}$ receiver. As a result, one cannot assume that the CSI remains unchanged for many iterative updates. However, the convergence behavior of the IWFA and GPA under time-varying channels is in general unknown.

Throughout this paper, we assume that the channel coefficients $\left\{h_{kk}^{(s)}, \forall~ s \in \mathcal{S}\right\}$ are known perfectly at both ends of the $k^{th}$ link, and the $k^{th}$ receiver measures local total received power $\big\{ \sigma^{2} +
\sum_{j = 1}^{K} g_{kj}^{(s)}(n)p_{j}^{(s)}(n)$, $\forall~ s\ \in \mathcal{S}   \big\}$ and feeds back this information to the  $k^{th}$ transmitter at the end of each time-slot, $\forall~ k \in \mathcal{K}$. \textcolor{black}{Table \ref{table:notations} summarizes the main notations used in the paper.}

\section{Switched System Modeling and Convergence Behavior Analysis}\label{sec:region_stability}
In general, the transient behavior of an iterative algorithm can be characterized by an {\em algorithm trajectory} of an associated nonlinear dynamic system and the NE is the associated \emph{equilibrium point} of the nonlinear dynamic system. Since the NE is a function of the CSI, a quasi-static CSI corresponds to a nonlinear system with \emph{static} equilibrium point. The convergence behavior of the iterative algorithm can be visualized as the algorithm trajectory converging to the equilibrium point as illustrated in Fig. \ref{fig:trajectory}. On the other hand, time-varying random CSI corresponds to a randomly moving NE (or a randomly moving equilibrium point) and the convergence behavior of the algorithm can be visualized as how well the algorithm trajectory could track the moving equilibrium point of the nonlinear system. In this section, we shall first propose a novel distributive {\em scaled gradient} power control algorithm for the interference {\em Game $\mathcal{G}$} under time-varying CSI. Based on this, we shall utilize the {\em randomly switched system} nonlinear control theory to analyze the convergence behavior as well as the tracking errors of the proposed scaled gradient projection algorithm.

\subsection{Distributed Scaled Gradient Projection Algorithm}\label{subsec:gradient_play}
The existence and uniqueness of NE of {\em Game $\mathcal{G}$} under \emph{quasi-static} channels has been extensively studied (e.g., see \cite{   Scutari2008_I, Scutari2008II , Yu2002} and references therein.). The solution set of {\em Game $\mathcal{G}$} is nonempty under any channel conditions, while the uniqueness of NE depends on the channel coefficients of the whole network \cite{Scutari2008_I , Yu2002}. As a result, with the time-varying channels, the NE of {\em Game $\mathcal{G}$} changes with time. We thus propose a  distributed scaled gradient projection algorithm (DSGPA) to \emph{track the moving NE} of {\em Game $\mathcal{G}$} in the time-varying interference networks.

Different from the gradient projection algorithms proposed in \cite{Scutari2008II, Xi2008} under quasi-static channels, where the scaling matrices are constant, the \emph{scaling matrices} in our proposed DSGPA can be adaptive to the CSI of the current update. Moreover, the scaling matrices could be optimized to  minimize the \emph{tracking error}, which will be detailed in section \ref{sec:mse_opt}.  Before summarizing the main algorithm, we first form the Lagrangian \cite{Boyd2004} of the optimization problem \eqref{eqn:game_formulation}:
\begin{IEEEeqnarray}{rCl} \label{eqn:LMk}
\mathcal{L}_{k}\left(\mathbf{p}_{k}(n, \lambda_{k}(n)\right) = C_{k}(\mathbf{p}_{k}(n), \mathbf{p}_{-k}(n)) + \lambda_{k}(n)\left(P_{k,max} - \mathbf{1}^{T}\mathbf{p}_{k}(n)\right), \forall k \in \mathcal{K},
\end{IEEEeqnarray}
where $\lambda_{k}(n) \geq 0$ is Lagrangian multiplier associated with the sum-power constraint at the $k^{th}$ transmitter. As has been established in \cite{Arrow1958}, solving the concave programming defined in equation \eqref{eqn:game_formulation}  and finding the saddle point\footnote{ A point $\left(\mathbf{p}_{k}^{*}(n), \lambda_{k}^{*}(n)\right)$ is a saddle point of the Lagrangian
$\mathcal{L}_{k}\left(\mathbf{p}_{k}(n), \lambda_{k}(n)\right)$ if
\begin{IEEEeqnarray}{rCl} \label{eqn:saddle_point_def}
\mathcal{L}_{k}\left(\mathbf{p}_{k}(n), \lambda_{k}^{*}(n)\right) \leq \mathcal{L}_{k}\left(\mathbf{p}_{k}^{*}(n), \lambda_{k}^{*}(n)\right) \leq \mathcal{L}_{k}\left(\mathbf{p}_{k}^{*}(n), \lambda_{k}(n)\right).\end{IEEEeqnarray}
} of the Lagrangian
$\mathcal{L}_{k}\left(\mathbf{p}_{k}(n), \lambda_{k}(n)\right)$ is equivalent. We thus adopt the \emph{primal-dual gradient method} \cite{Arrow1958} to find the saddle point of the Lagrangian $\mathcal{L}_{k}\left(\mathbf{p}_{k}(n, \lambda_{k}(n)\right)$. The proposed DSGPA is summarized in {\em Algorithm \ref{alg:sgpa}}.
\begin{algorithm}
\caption{Distributed Scaled Gradient Projection Algorithm (DSGPA)}
\label{alg:sgpa}
\begin{itemize}
\item \textbf{Initialization}: \\
Set $\mathbf{p}_{k}(0)$ to be any feasible power allocation vector;
$\lambda_{k}(0)$ to be any positive number; and the scaling matrix
$\mathbf{D}_{k}(0)$ to be $\mathbf{I}_{N_{F}}$, $\forall~  k \in
\mathcal{K}$;

\item \textbf{Updating the Lagrangian Multiplier $\lambda_{k}(n)$}:\\
\emph{At each time-slot $n$} ($n\geq 1$), after receiving the \emph{received power profile} $\bm{\rho}_{k}(n)$ (see equation \eqref{eqn:rpwr} below) fed back by the ${k^{th}}$ receiver, the ${k^{th}}$ transmitter
updates its Lagrangian Multiplier $\lambda_{k}(n) \in
\mathbb{R}_{+}$ according to
\begin{IEEEeqnarray}{rCl} \label{eqn:alg_LMk}
\lambda_{k}(n+1) =  \left[ \lambda_{k}(n) -
\alpha\left(P_{k,max}-\mathbf{1}^{T}\mathbf{p}_{k}(n)\right)\right]^{+};
\end{IEEEeqnarray}
\item \textbf{Updating the power allocation vector $\mathbf{p}_{k}(n)$}: \\
The ${k^{th}}$ transmitter updates its \emph{power allocation
vector} $\mathbf{p}_{k}(n) \in \mathbb{R}_{+}^{N_{F}} $ according to
\begin{IEEEeqnarray}{rCl} \label{eqn:update_pin}
\mathbf{p}_{k}(n+1) =  \left[ \mathbf{p}_{k}(n) +
\mathbf{D}_{k}^{-1}(n)\mathbf{f}_{k}(n) \right]^{+}.
\end{IEEEeqnarray}
\end{itemize}
\end{algorithm}

In {\em Algorithm \ref{alg:sgpa}}:
\begin{itemize}
\item[--]\textcolor{black}{the $k^{th}$ transmitter needs the channel gain $\left\{ h_{kk}^{(s)}, \forall s \in \mathcal{S}\right\}$ feedback from its intended receiver $k$};
\item[--] the \emph{received power profile} $\bm\rho_{k}(n)$ is given by
\begin{IEEEeqnarray}{rCl} \label{eqn:rpwr}
\bm\rho_{k}(n) = \left[\rho_{k}^{(1)}(n) \:\ \rho_{k}^{(2)}(n) \:\ \cdots \:\ \rho_{k}^{(N_{F})}(n)\right]^{T}, \forall~ k \in \mathcal{K},
\end{IEEEeqnarray}
with $\rho_{k}^{(s)}(n) = \sigma^{2} +
\sum_{j = 1}^{K} g_{kj}^{(s)}(n)p_{j}^{(s)}(n)$ denoting the total received power on the $s^{th}$ subcarrier  at the $k^{th}$ receiver, $\forall~ s \in \mathcal{S}$, $k \in \mathcal{K}$;
\item[--] $\alpha$ is a stepsize, which can be computed efficiently via the {\em exact line search} method \cite{Boyd2004}, as the Lagrangian given by \eqref{eqn:LMk} is linear in $\lambda_{k}(n)$;
\item[--] $\mathbf{D}_{k}^{-1}(n)$ denotes the \emph{symmetric positive definite} scaling matrix, defined in equation \eqref{eqn:Di_soln}; \textcolor{black}{to facilitate distributive computation of the dynamic scaling matrices at each node, we impose a block diagonal structure\footnote{\textcolor{black}{We shall illustrate in Section \ref{sec:num_res_disc} that there is only very small performance penalty associated with the block diagonal scaling matrix.}}  on the scaling matrix $\mathbf{D}_{k}^{-1}(n)$;}
\item[--] $\mathbf{f}_{k}(n)$ is defined as the gradient of $\mathcal{L}_{k}\left(\mathbf{p}_{k}(n, \lambda_{k}(n)\right)$ w.r.t. $\mathbf{p}_{k}(n)$ evaluated at $\mathbf{p}(n)$, i.e.,
    \begin{IEEEeqnarray}{rCl} \label{eqn:gradient}
  \mathbf{f}_{k}(n) \triangleq \left[ \frac{g_{kk}^{(1)}(n)}{\rho_{k}^{(1)}(n)} \:\ \frac{g_{kk}^{(2)}(n)}{\rho_{k}^{(2)}(n)} \:\ \cdots \:\ \frac{g_{kk}^{(N_{F})}(n)}{\rho_{k}^{(N_{F})}(n)}\right]^{T} - \lambda_{k}(n+1)\mathbf{1};
\end{IEEEeqnarray}
\item[--] $\left[\mathbf{a}\right]^{+} = \max\{\mathbf{a}, \mathbf{0}\}$, which shall be understood componentwisely.
\end{itemize}
\begin{Rem}
In {\em Algorithm \ref{alg:sgpa}}, at each time slot $n$, the ${k^{th}}$ transmitter updates its power allocation vector based on the \emph{local information} $\left\{g_{kk}^{(s)}(n), \forall~ s \in \mathcal{K}\right\}$, $\mathbf{p}_{k}(n)$, and the \emph{local} receive power profile $\bm\rho_{k}(n)$ of the $k^{th}$ receiver. Therefore, {\em algorithm \ref{alg:sgpa}} can be implemented distributively, without requiring any global information of the network.
\end{Rem}


\begin{Rem}
As the FSMC $\left\{\mathbf{h}(n)\right\}$ jumps from one state to another randomly, the NE of {\em Game $\mathcal{G}$} also changes with time randomly. As a result, {\em Algorithm \ref{alg:sgpa}} would not converge to a single point but rather designed to \emph{track} the moving NE of {\em Game $\mathcal{G}$}. The \emph{tracking performance} of {\em Algorithm \ref{alg:sgpa}} is the focus of the rest of this paper.
\end{Rem}

\subsection{Randomly Switched System Modeling}\label{subsec:nash_equilibria}
Randomly switched systems are \emph{piecewise deterministic} stochastic systems, i.e., between any two consecutive \emph{switching instants}, the dynamics are deterministic \cite{Liberzon2003}. Formally, a \emph{discrete-time randomly switched system} is defined as follows \cite{Liberzon2003}.

\begin{Def} [Discrete-time Randomly Switched System] A Discrete-time  \emph{Randomly switched system} consists of a family of \emph{subsystems}, and a random \emph{switching signal} that specifies the \emph{active} subsystem at every time-slot. Mathematically,
\begin{IEEEeqnarray}{rCl} \label{eqn:Def_S}
\mathbf{x}(n + 1) = \mathscr{F}_{u}\left( \mathbf{x}(n) \right), \mbox{when }\tau\left(n\right) = u \in \mathcal{U} = \{1, 2, \cdots, U\},
\end{IEEEeqnarray}
where $\mathbf{x}(n)$ denotes the system state; $\mathscr{F}_{u}\left( \mathbf{x}(n) \right)$ denotes the $u^{th}$ \emph{subsystem}; and $\tau\left(n\right)$ is the switching signal with state space $\mathcal{U}$. ~\hfill\IEEEQEDclosed
\end{Def}

For the FSMC $\left\{\mathbf{h}(n)\right\}$, the channel process stays in a state $\mathbf{h}_{q}$ for a {\em random sojourn time} of $N_{q}$ time-slots, and then jumps to another state randomly. During the $N_{q}$ time-slots, the channel coefficients $\left\{\mathbf{h}(n)\right\}$ remain constant (i.e., $\mathbf{h}(n) = \mathbf{h}_{q}$) and the system is deterministic. We thus can model the time-varying interference network embedded with the dynamics of the proposed DSGPA as a \emph{randomly switched system}, with the FSMC $\left\{\mathbf{h}(n)\right\}$ being the \emph{switching signal} and the channel state $\mathbf{h}_{q}$ corresponding to the $q^{th}$ \emph{subsystem} \cite{Liberzon2003}, for all $q \in \mathcal{Q}$.

To obtain the dynamics of the randomly switched system model of the power control game explicitly,  we rewrite the $K$ block-component iterations given by equation \eqref{eqn:update_pin}   in {\em Algorithm \ref{alg:sgpa}} into one vector form:
\begin{IEEEeqnarray}{rCl} \label{eqn:system_iteration}
\mathbf{p}(n+1) = \left[ \mathbf{p}(n) + \mathbf{D}^{-1}(n)
\mathbf{f}(n) \right]^{+}  \triangleq \mathscr{T}_{q}\left(\mathbf{p}(n)\right), \mbox{when } \mathbf{h}(n) = \mathbf{h}_{q} \in \mathcal{H},
\end{IEEEeqnarray}
where
\begin{itemize}
\item[--] $\mathbf{p}(n)$ denotes the transmit power profile of all the players, i.e.,
\begin{IEEEeqnarray}{rCl} \label{eqn:pn}
\mathbf{p}(n) = \left[ \mathbf{p}_{1}^{T}(n) \:\
\mathbf{p}_{2}^{T}(n) \:\ \cdots \:\ \mathbf{p}_{K}^{T}(n)
\right]^{T};
\end{IEEEeqnarray}
\item[--] $\mathbf{D}^{-1}(n) = blkdlg\left\{ \mathbf{D}_{1}^{-1}(n),
\mathbf{D}_{2}^{-1}(n), \cdots, \mathbf{D}_{K}^{-1}(n)  \right\} \in \mathbb{C}^{KN_{F} \times KN_{F}}$ denotes the\emph{ block-diagonal} scaling matrix consisting of the $K$ scaling matrices
$\left\{ \mathbf{D}_{k}^{-1}(n), \forall~ k \in \mathcal{K} \right\}$;
\item[--] $\mathbf{f}(n) \in \mathcal{R}^{KN_{F}}$ denotes the collection of the $K$ gradient functions $\left\{\mathbf{f}_{k}(n), \forall~ k \in \mathcal{K}\right\}$, i.e.,
\begin{IEEEeqnarray}{rCl} \label{eqn:gradient_collection}
\mathbf{f}(n) = \left[\mathbf{f}_{1}^{T}(n) \:\ \mathbf{f}_{2}^{T}(n) \:\ \cdots \:\ \mathbf{f}_{K}^{T}(n) \right]^{T};
\end{IEEEeqnarray}
\item[--] $\mathscr{T}_{q}\left(\mathbf{p}(n)\right)$ denotes the dynamics of the $q^{th}$ subsystem, $\forall~ q \in \mathcal{Q}$.
\end{itemize}
\begin{Rem}
In the above randomly switched system modeling,  $\mathbf{p}(n)$ is the system state vector, and the $Q$ states of the FMSC $\left\{\mathbf{h}(n)\right\}$ correspond to the $Q$ \emph{subsystems} $\big\{\mathscr{T}_{1}\big(\mathbf{p}(n)\big), \mathscr{T}_{2}\big(\mathbf{p}(n)\big),$ $\cdots, \mathscr{T}_{Q}\big(\mathbf{p}(n)\big) \big\}$. As the FMSC $\left\{\mathbf{h}(n)\right\}$ jumps between different states, the switched system \eqref{eqn:system_iteration} switches between different subsystems.
\end{Rem}

To simplify the analysis in the sequel, we make the following assumptions.
\begin{Ass} [Existence of NE] \label{ass:unique_NE} We assume that the interference {\em Game $\mathcal{G}$ } in \eqref{eqn:game_formulation} has \textcolor{black}{a} unique NE for all $q \in \mathcal{Q}$. In other words, using the sufficient conditions for the existence and uniqueness of NE of {\em Game $\mathcal{G}$ } given in \cite{Scutari2008_I, Yu2002}, we assume that
\begin{IEEEeqnarray}{rCl} \label{eqn:sufficient_NE}
\max_{s \in \mathcal{S}} \frac{g_{kj}^{(s)}(n)P_{j, max}}{g_{kk}^{(s)}(n)P_{k, max}} < \frac{1}{K-1}, \forall~ k, j \in \mathcal{K}, k \neq j.
\end{IEEEeqnarray}
\end{Ass}

\begin{Rem}
\textcolor{black}{Under {\em Assumption \ref{ass:unique_NE}}, the proposed distributed scaled gradient projection algorithm (DSGPA) converges linearly for each channel state \cite{Scutari2008_I}. In other words, the proposed DSGPA converges linearly if the channel is static (cardinality of the CSI space $\left|\mathcal{H}_{kj}^{(s)}\right| = \widetilde{Q}_{kj}^{(s)}=1$). }
\end{Rem}

\subsection{Convergence Analysis of the Proposed DSGPA}\label{subsec:region_convergence}
The \emph{region stability} is a widely used performance measure of iterative algorithms in time-varying environments, especially for switched and hybrid systems \cite{Podelski_Proof, Alpcan2004}.
When the CSI is time-varying (e.g., the FSMC model), the equilibrium point of the system (e.g., the NE) is also changing and hence, the {\em algorithm trajectory} of the iterative algorithm will not converge to a single point but rather a limit region as illustrated in Fig. \ref{fig:region_stability}. We shall formally define  \emph{region stability} below.
\begin{Def} [Region Stability of Switched Systems] \label{def:region_stability}
A discrete-time randomly switched system  with state vector
$\mathbf{p}(n)$ is said to be stable w.r.t. a \emph{limit region $\mathcal{L}$}, if for every trajectory $\mathbf{p}\left(n, \mathbf{p}(0)\right)$, there exists a point of
time $N_{0}\left(\mathbf{p}(0)\right)$ such that from then on, the
trajectory is always in the limit region $\mathcal{L}$. Mathematically,
\begin{IEEEeqnarray}{rCl} \label{eqn:def_region}
\forall~ \mathbf{p}\left(n, \mathbf{p}(0)\right), ~ \exists~ N_{0}\left(\mathbf{p}(0)\right) \mbox{ such that }
\mathbf{p}\left(n, \mathbf{p}(0)\right) \in \mathcal{E}, \forall~ n \geq N_{0}\left(\mathbf{p}(0)\right).
\end{IEEEeqnarray}
\end{Def}

Before proceeding further, we  introduce the following intermediate definitions.
\begin{Def} [Matrix-2 Norm] The \emph{matrix-2 norm}  $ \left \|\mathbf{A}\right \|_{2}$ is defined to be \cite{Bertsekas1989}
\begin{IEEEeqnarray}{rCl} \label{eqn:matrixnorm}
\left \|\mathbf{A}\right \|_{2} = \max_{ \left\{ \mathbf{x}:~ \left \|\mathbf{x}\right \|_{2}= 1  \right\} } \left \|\mathbf{A}\mathbf{x}\right \|_{2},
\end{IEEEeqnarray}
where the \emph{vector norm} $\left \|\mathbf{x}\right \|_{2}$ is defined as \cite{Bertsekas1989}: $\left \|\mathbf{x}\right \|_{2} = \sqrt{\mathbf{x}^{T}\mathbf{x}}$.
\end{Def}
\begin{Def} [Vector Block-maximum Norm] The vector \emph{block-maximum norm}  on the power profile $\mathbf{p}(n)$ is defined to be \cite{Bertsekas1989}
\begin{IEEEeqnarray}{rCl} \label{eqn:vectorblocknorm}
\left \|\mathbf{p}(n)\right \|_{\textmd{block}} \triangleq \max_{k \in \mathcal{K}} \left \|\mathbf{p}_{k}(n)\right \|_{2}.
\end{IEEEeqnarray}
\end{Def}

For ease of elaboration, we also introduce the following \emph{contraction modulus} \cite{Scutari2008II, Bertsekas1989}:
\begin{IEEEeqnarray}{l} \label{eqn:xiDin}
\beta_{k}\left(\mathbf{D}_{k}(n)\right) \triangleq \left\| \mathbf{I}_{N_{F}} +  \mathbf{D}_{k}^{-1}(n) \partial_{kk}^{2}C_{k}(n) \right \|_{2} + \sum_{j =1, j \neq k}^{K} \left\|\mathbf{D}_{k}^{-1}(n) \partial_{kj}^{2}C_{k}(n) \right\|_{2}, \forall~ k \in \mathcal{K},
\end{IEEEeqnarray}
where $ \partial_{kj}^{2}C_{k}(n) \in \mathbb{C}^{N_{F} \times N_{F}} $ denotes the second order partial derivative of $C_{k}(n)$ w.r.t. $\mathbf{p}_{j}(n)$ evaluated at $\mathbf{p}(n)$, i.e.,
\begin{IEEEeqnarray}{rCl} \label{eqn:delta_pi}
\partial_{kj}^{2}C_{k}(n) = \frac{\partial^{2}C_{k}(n)}{\partial \mathbf{p}_{k}(n) \partial \mathbf{p}_{j}(n)} = diag\left(\left[ \eta_{kj}^{(1)}(n)  \:\ \eta_{kj}^{(2)}(n) \:\
\cdots \:\ \eta_{kj}^{(N_{F})}(n) \right]^{T} \right), \forall~ k, j \in \mathcal{K},
\end{IEEEeqnarray}
where $\eta_{kj}^{(s)}(n) = - \frac{ g_{kk}^{(s)}(n)g_{kj}^{(s)}(n) }
{\left(\rho_{k}^{(s)}(n)\right)^{2} }, \forall~ s \in \mathcal{S}$.

We now summarize the  region stability property of the proposed DSGPA  in the following theorem.
\begin{Thm} [Region Stability of DSGPA]  \label{thm:region_stability} Under the conditions that $\big\{\beta_{k}\left(\mathbf{D}_{k}(n)\right) < 1, \forall~ k \in \mathcal{K}, \forall~ n \geq 1\big\}$, for \textcolor{black}{sufficiently} large $n$, the probability that the iterates $\mathbf{p}(n)$ generated by \eqref{eqn:system_iteration} being outside the limit region $\mathcal{L}$, can be upper bounded by
\begin{IEEEeqnarray}{rCl} \label{eqn:region_prob}
P_{Region} \triangleq \lim_{n \rightarrow +\infty }\mbox{Pr}\left\{\mathbf{p}(n) \notin \mathcal{L}\right\} \leq \min_{}\left\{1, \frac{\beta}{(1-\beta)\overline{N}} \right\},
\end{IEEEeqnarray}
where $\beta = \max_{\{\forall~k \in \mathcal{K}, \forall~ n \geq 1\}} \left\{\beta_{k}\left(\mathbf{D}_{k}(n)\right) \right\}$ is the maximum contraction modulus; $\overline{N} = \frac{1}{1-\nu^{K^{2}N_{F}}}$ is the \emph{average sojourn time} of  the FSMC $\left\{\mathbf{h}(n)\right\}$; the limit region $\mathcal{L}$ is given by
\begin{IEEEeqnarray}{rCl} \label{eqn:limit_region}
\mathcal{L} = \bigcup_{q = 1}^{Q}\mathcal{L}_{q} = \bigcup_{q = 1}^{Q}\left\{ \mathbf{p} \big| \left\|\mathbf{p} - \bar{\mathbf{p}}^{(q)} \right\|_{\textmd{block} } \leq \delta \right\}
\end{IEEEeqnarray}
where $\delta = \max_{\left\{\forall q,r  \in \mathcal{Q}, q \neq r \right\} } \left \|  \bar{\mathbf{p}}^{(q)} - \bar{\mathbf{p}}^{(r)} \right \|_{\textmd{block}}$ denotes the maximum distance between two NEs in $\big\{ \bar{\mathbf{p}}^{(q)}, \forall~ q \in \mathcal{Q} \big\}$, \textcolor{black}{with $\bar{\mathbf{p}}^{(q)}$ denoting the NE of  channel state $q$. The limit region is the convex hull \cite{Boyd2004} of the NEs $\big\{ \bar{\mathbf{p}}^{(q)}, \forall~ q \in \mathcal{Q} \big\}$ corresponding to different channel states. As a result, the limit region is a polyhedron with diameter $\delta \approx \frac{P_{max}}{4}$, where $P_{max} \triangleq \max_{k}P_{k,max}$.}
\end{Thm}
\begin{proof}
Please refer to Appendix \ref{app:region_stability} for the proof.
\end{proof}
\begin{Rem}
In {\em Theorem \ref{thm:region_stability}}, $\overline{N}$ can be thought as an indicator of the channel fading rate. The larger the $\overline{N}$ is, i.e., the channel changes more slowly, the smaller the $P_{Region}$ is. In particular, we have $P_{Region} = \mathcal{O}\left( 1\big/\overline{N} \right)$
\end{Rem}

\subsection{Asymptotic Order of Growth of the Tracking Errors}\label{subsec:scaling_matrix}
Steady state tracking error is the main concern when designing an iterative algorithm in a time-varying environment. Here, we consider the {\em expected-absolute-error} (EAE) and the \emph{mean-square-error} (MSE), which are formally defined as follows.
\textcolor{black}{\begin{Def} [EAE and MSE] \label{def:bias}
The EAE (or MSE) is defined to be the expectation of the distance (or \emph{squared} distance) between the iterate $\mathbf{p}(n)$ and the  corresponding NE $\mathbf{\bar{p}}^{(q)}$, i.e.,
\begin{IEEEeqnarray}{rCl} \label{eqn:def_mse}
\mbox{EAE}\left(\mathbf{p}(n)\right) = \mathbb{E}_{\{\mathbf{h}(n)\}} \left\{ \left\|    \mathbf{p}(n) - \mathbf{\bar{p}}^{(q)} \right\|_{\textmd{block}} \right\}, \\
\mbox{MSE}\left(\mathbf{p}(n)\right) = \mathbb{E}_{\{\mathbf{h}(n)\}}\left\{\left\| \mathbf{p}(n) - \mathbf{\bar{p}}^{(q)} \right\|_{\textmd{block}}^{2}\right\},
\end{IEEEeqnarray}
where the expectation shall be taken over the stationary distribution of the FMSC $\left\{\mathbf{h}(n)\right\}$. ~\hfill\IEEEQEDclosed
\end{Def}}
%

The asymptotic order of growth  of the expected-absolute-error $\mbox{EAE}\left(\mathbf{p}(n)\right)$ and the mean-square-error $\mbox{MSE}\left(\mathbf{p}(n)\right)$ are summarized in the following theorem.
\begin{Thm} [Asymptotic Order of Growth of EAE and MSE] \label{thm:UBD_MSE} Under the conditions that $\big\{\beta_{k}\left(\mathbf{D}_{k}(n)\right) < 1, \forall~ k \in \mathcal{K}, \forall~ n \geq 1\big\}$,  the order of growth of the expected-absolute-error $\mbox{EAE}\left(\mathbf{p}(n)\right)$ and mean-square-error $\mbox{MSE}\left(\mathbf{p}(n)\right)$ during steady state are given by:
\begin{IEEEeqnarray}{rCl} \label{eqn:UBD_MSE}
\mbox{EAE}\left(\mathbf{p}(n)\right) = \mathcal{O}\left(  \frac{ \beta}{(1-\beta)\overline{N}} \right), \mbox{ and }~
\mbox{MSE}\left(\mathbf{p}(n)\right) = \mathcal{O}\left( \frac{ \beta^{2}\left(2\beta + (1-\beta) \overline{N}\right) }{(1-\beta^{2})(1-\beta)\overline{N}^{2}} \right),
\end{IEEEeqnarray}
where $\beta = \max_{\{\forall~k \in \mathcal{K}, \forall~ n \geq 1\}} \left\{\beta_{k}\left(\mathbf{D}_{k}(n)\right) \right\}$ is the maximum contraction modulus.
\end{Thm}
\begin{proof}
Please refer to Appendix \ref{app:UBD_MSE} for the proof.
\end{proof}
\begin{Rem}
The expressions  of the tracking errors EAE and MSE given in {\em Theorem \ref{thm:UBD_MSE}} depend on the average sojourn time $\overline{N}$, which can be thought as an indicator of the \emph{channel fading rate}. The larger the $\overline{N}$ is, i.e., the more slowly the channel changes, the smaller are the tracking errors. Particularly, we have $\mbox{EAE}\left(\mathbf{p}(n)\right) = \mathcal{O}\left( 1\big/\overline{N} \right)$ and $\mbox{MSE}\left(\mathbf{p}(n)\right) = \mathcal{O}\left( 1\big/\overline{N} \right)$.
\end{Rem}

\section{Tracking Error Optimization}\label{sec:mse_opt}
In the previous section, we have established the region stability property of the proposed DSGPA and derived the order of growth of the tracking errors.  In this section, we shall design the scaling matrices $\left\{\mathbf{D}_{k}(n), \forall~ k \in \mathcal{K}, \forall~ n \geq 1 \right\}$ to minimize the tracking errors.  Specifically,  we shall first construct a {\em dominated error process}.  Based on  that, we shall optimize the scaling matrices via the \emph{Markov decision process} (MDP) modeling and show that the optimal scaling matrixes can be computed distributively.

\subsection{Tracking Error Optimization}\label{subsec:scaling_matrix}
In this section, we shall derive an {\em tracking error optimal scaling matrices} for the proposed DSGPA. \textcolor{black}{For ease of elaboration, we first introduce the the notion of \emph{stage} below.
\begin{Def} [Stage] A \emph{Stage} is defined as the time-span, for which the channel fading process $\left\{\mathbf{h}(t)\right\}$ remains unchanged (i.e., stays at the same channel state), as illustrated in Fig. \ref{fig:stage}. ~\hfill\IEEEQEDclosed
\end{Def}}

We next proceed to construct a {\em dominated error process} $\{\widetilde{e}(m)\}$ defined as:
\begin{IEEEeqnarray}{rCl} \label{eqn:dominanting_error}
\widetilde{e}(m+1) = \widetilde{e}(m)\phi_{m}^{N_{m}} + \delta_{m, m+1}, \forall~ m = 1, 2, \cdots ,
\end{IEEEeqnarray}
where $\phi_{m}$ denotes the worse case contraction modulus of all the transmitters at the $m^{th}$ stage; $\delta_{m, m+1} = \left \|\mathbf{\bar{p}}(m) - \mathbf{\bar{p}}(m+1) \right \|_{\textmd{block}}$  denotes the distance between the NE  $\mathbf{\bar{p}}(m)$  of the $m^{th}$ stage and the NE  $\mathbf{\bar{p}}(m+1)$ of the  $(m+1)^{th}$  stage.
The  \emph{dominating error process} $\{\widetilde{e}(m)\}$ has the following properties.
\begin{Lem} [Property of Dominating Error Process]  \label{lem:dominanting_error}
The dominating error process $\{\widetilde{e}(m)\}$ defined in equation \eqref{eqn:dominanting_error} is an upper bound of the distance (w.r.t. the vector block-maximum norm) between the algorithm trajectory $\{\mathbf{p}(n)\}$ and the NE  at the beginning of the $m^{th}$ stage at the steady state, i.e.
\begin{IEEEeqnarray}{rCl} \label{eqn:almost_sure}
0 \leq e(m) \leq \widetilde{e}(m) \leq \frac{\delta\beta}{1-\beta}, \mbox{ almost surely (a.s.)}, \forall~ m = 1,2, \cdots ,
\end{IEEEeqnarray}
where $\{e(m)\}$ denotes the actual \emph{initial} error of each stage, and we choose $\widetilde{e}(1) = e(1)$.
\end{Lem}
\begin{proof}
Please refer to Appendix \ref{app:dominanting_error} for the proof.
\end{proof}

Let $\bm{\chi}(n) = \left(\widehat{e}(n) , \mathbf{h}(n) \right) \in \mathbb{R}_{+} \bigotimes \mathbb{C}^{K^{2}N_{F} \times 1}$ denote the system state  at the $n^{th}$ time-slot, \textcolor{black}{where it is assumed that $\widehat{e}(n)= \widetilde{e}(m)$ when the $n^{th}$ time-slot is in the span of the $m^{th}$ stage.} For  a given system state realization $\bm{\chi}(n)$, the transmitters adjust the scaling matrix action $\mathbf{D}(n)$ according to a stationary \emph{scaling matrix control policy} $\pi = \mathbf{D}\left(\bm{\chi}(n)\right)$ defined below.
\begin{Def} [Stationary Scaling Matrix Control Policy] A stationary scaling matrix control policy  $\pi:\mathbb{R}_{+} \bigotimes \mathbb{C}^{K^{2}N_{F} \times 1} \rightarrow \mathbb{C}^{KN_{F} \times KN_{F}}$ is defined as the mapping from the currently observed system state $\bm{\chi}(n)$ to a scaling matrix action $\mathbf{D}(n)$. ~\hfill\IEEEQEDclosed
\end{Def}

Using {\em Lemma \ref{lem:dominanting_error}}, we shall derive an optimal scaling matrix control policy w.r.t. an average tracking error upper bound (represented by $\widehat{e}(n)$). We further assume that the \emph{dominating error process} $\{\widehat{e}(n)\}$ admits finite values\footnote{Since $\{\widehat{e}(n)\}$ is bounded, we can always set a realization of $\{\widehat{e}(n)\}$ to a larger nearest integer, which results in a finite integer-valued random process $\{\widehat{e}(n)\}$.} in  $\mathcal{E} = \{\bar{e}_{1}, \bar{e}_{2}, \cdots,$ $\bar{e}_{L}\}$. As a result, given a stationary scaling matrix control policy $\pi$,  $\left\{\bm{\chi}(n)\right\}$ is an induced Markov Chain, and the transition probability of $\bm{\chi}(n)$ is given by:
\begin{IEEEeqnarray}{l} \label{eqn:MDPkernel}
 Pr\left\{ \bm{\chi}(n+1) = (\bar{e}_{l}, \mathbf{h}_{r}) \big| \bm{\chi}(n) = (\bar{e}_{i}, \mathbf{h}_{q}),   \pi\left( \bm{\chi}(n)\right)  \right\} =  \\ \notag
 T_{qr}\left(1-\nu^{K^{2}N_{F}}\right)\left( \frac{ \bar{e}_{l} - \delta_{qr} }{ \bar{e}_{i} } \right)^{
K^{2}N_{F}\log_{ \beta \left( \bm{\chi(n) } \right) } \nu}, \forall~ 1 \leq l,  i \leq L, 1 \leq q,  r \leq Q,
\end{IEEEeqnarray}
where $T_{qr} = \left[\mathbf{T}\right]_{qr}$ is the transition probability from state $q$ to state $r$ of the FSMC $\left\{\mathbf{h}(n)\right\}$; $\delta_{qr} = \left \|  \bar{\mathbf{p}}^{(q)} - \bar{\mathbf{p}}^{(r)} \right \|_{\textmd{block}}$ is the  distance between the NE  $\mathbf{\bar{p}}^{(q)}$  of the $m^{th}$ stage and the NE  $\mathbf{\bar{p}}^{(r)}$ of the  $(m+1)^{th}$  stage;  and $\beta \left( \bm{\chi}(n)\right)  = \max_{k \in \mathcal{K}} \left\{  \beta_{k}\left(\mathbf{D}_{k} (n)\right) \right\}$  is the maximum contraction modulus at time-slot $n$.

The tracking error optimization problem is formally given below.
\begin{Prob} [Tracking Error Control Problem] \label{prob:MDPformulation}
To minimize the average tracking error, the adaptive scaling matrix control policy $\pi^{*}$ is given by
\begin{IEEEeqnarray}{rCl} \label{eqn:MDPformulation}
\pi^{*} = \argmin{\pi}{J^{\pi} } ,  \mbox{ with } J^{\pi} = \limsup_{N \rightarrow + \infty}\frac{1}{N}\sum_{n=1}^{N}\mathbb{E}^{\pi}\left\{ g\left(\widehat{e}(n) \right) \right\},
\end{IEEEeqnarray}
where $g\left(\widehat{e}(n)\right)$ is an increasing function of $\{\widehat{e}(n)\}$, which measures the per-stage tracking error and $\mathbb{E}^{\pi}$ denotes the expectation w.r.t. the induced measure (induced by the control policy $\pi$).
\end{Prob}

In general, the optimization problem in \eqref{eqn:MDPformulation} is very difficult to solve due to the huge dimensions of variables (control policy) involved as well as difficulty to express the optimization objective function $\mathbb{E}^{\pi}\left\{ g\left(\widehat{e}(n) \right) \right\}$ explicitly as the variable $\pi$. Yet, utilizing the special structure of the transition kernel in \eqref{eqn:MDPkernel}, the solution of {\em Problem \ref{prob:MDPformulation}} is summarized in the following theorem, \textcolor{black}{which shows that the adaptive scaling matrices $\left\{\mathbf{D}_{k}(n), \forall k \in \mathcal{K}\right\}$ can be computed independently in every time-slot.}
\begin{Thm} [Solution of the Tracking Error Control Problem] The optimal scaling matrix control policy for {\em Problem \ref{prob:MDPformulation}} is  given by the solution of the following optimization problem. \label{thm:solution_mdp}
\begin{IEEEeqnarray}{rCl} \label{eqn:mse_prob}
\begin{array}{cc}
\begin{array}{c}
 \minimize{\mathbf{D}(n)}  \\
\mbox{subject to} \\ \mbox{} \end{array} &
\begin{array}{l}
\max_{k \in \mathcal{K}}\beta_{k}\left(\mathbf{D}_{k}(n)\right)  \\ \max_{k \in \mathcal{K}}\beta_{k}\left(\mathbf{D}_{k}(n)\right)  < 1, \\ \mathbf{D}(n) \succ 0
 \end{array}
\end{array}
\end{IEEEeqnarray}
where $\beta_{k}\left(\mathbf{D}_{k}(n)\right)$ denotes the contraction modulus, which is defined in equation \eqref{eqn:xiDin}.
\end{Thm}
\begin{proof}
Please refer to Appendix \ref{app:solution_mdp} for the proof.
\end{proof}

\subsection{Distributed Implementation of the Optimal Scaling Matrices}\label{subsec:scaling_matrix}
By exploiting  the \emph{block-diagonal} structure of $\mathbf{D}(n)$, the {\em optimization problem \eqref{eqn:mse_prob}} can be naturally decoupled into $K$ subproblems and solved \emph{distributively} at the $K$  transmitters. The subproblem that needs to be solved at the $k^{th}$ transmitter can be formulated as follows.
\begin{Prob} [Subproblem of Scaling Matrix Optimization] \label{prob:mse_submini}
The optimal scaling matrix $\mathbf{D}_{k}(n)$ at the $k^{th}$ transmitter for minimizing the tracking errors is given by the solution of the following problem, $\forall~ k \in \mathcal{K}$.
\begin{IEEEeqnarray}{rCl} \label{eqn:mse_subprob}
\begin{array}{cc}
\begin{array}{c}
 \minimize{\mathbf{D}_{k}(n)}  \\
\mbox{subject to} \\ \mbox{} \end{array} &
\begin{array}{l}
\beta_{k}\left(\mathbf{D}_{k}(n)\right)  \\ \beta_{k}\left(\mathbf{D}_{k}(n)\right)  < 1, \\ \mathbf{D}_{k}(n) \succ 0. \end{array}
\end{array}
\end{IEEEeqnarray}
\end{Prob}

Since the objective function $\beta_{k}\left(\mathbf{D}_{k}(n)\right)$ of {\em Problem \ref{prob:mse_submini}}  consists of  linear functions of the positive definite matrix $\mathbf{D}_{k}(n)$ and sum of matrix norms,  {\em Problem \ref{prob:mse_submini}} is a convex optimization problem \cite{Boyd2004}.  However, the objective function $\beta_{k}\left(\mathbf{D}_{k}(n)\right)$ is \emph{not} differentiable. To find a closed-form solution of {\em Problem \ref{prob:mse_submini}}, we need  the following intermediate results.

\begin{Lem} [Objective Function of Problem \ref{prob:mse_submini}] \label{lem:Trans_sumprob}
The objective function $\beta_{k}\left(\mathbf{D}_{k}(n)\right)$ of {\em Problem \ref{prob:mse_submini}} can be rewritten as
\begin{IEEEeqnarray}{rCl} \label{eqn:xiDin_Trans}
\beta_{k}\left(\mathbf{D}_{k}(n)\right) = \left\| \mathbf{I}_{N_{F}} +  \mathbf{D}_{k}^{-1}(n) \partial_{kk}^{2}C_{k}(n) \right \|_{2} + \sum_{j =1, j \neq k}^{K} \left\|\mathbf{D}_{k}^{-1}(n) \partial_{kj}^{2}C_{k}(n) \right\|_{2}, \forall~ k \in \mathcal{K},
\end{IEEEeqnarray}
and can be lower bounded by
\begin{IEEEeqnarray}{rCl} \label{eqn:lowerBD_G}
     \beta_{k}\left(\mathbf{D}_{k}(n)\right) \geq \sum_{j =1, j \neq k}^{K} \max_{s \in \mathcal{S}}\frac{g_{kj}^{(s)}(n)}{g_{kk}^{(s)}(n)}, \forall~ k \in \mathcal{K}, \forall~ \mathbf{p}(n) \in \Omega.
     \end{IEEEeqnarray}
\end{Lem}
\begin{proof}
Please refer to Appendix \ref{app:Trans_sumprob} for the proof.
\end{proof}

By virtue of {\em Lemma \ref{lem:Trans_sumprob}}, we can get a closed-form solution for {\em Problem \ref{prob:mse_submini}}. We summarize the main results of this section into the following theorem.
\begin{Thm} [Optimal Solution of the Tracking Error Control Problem] \label{thm:optimal_soln}
The optimal solution $\mathbf{D}_{k}(n) $ of  {\em Problem \ref{prob:mse_submini}} is given by
\begin{IEEEeqnarray}{rCl} \label{eqn:Di_soln}
\mathbf{D}_{k}(n) = -\partial_{kk}^{2}C_{k}(n), \forall~ k \in \mathcal{K}.
\end{IEEEeqnarray}
\end{Thm}
\begin{proof}
Please refer to Appendix \ref{app:optimal_soln} for the proof.
\end{proof}

From {\em Theorem \ref{thm:optimal_soln}} we know that the optimal scaling matrix $\mathbf{D}_{k}^{-1}(n)$ is given by the \emph{minus inverse of the second order partial derivative} of the capacity function $C_{k}(n)$  w.r.t. the power allocation vector $\mathbf{p}_{k}(n)$, $\forall~ k \in \mathcal{K}$. Therefore,  the scaling matrix $\mathbf{D}_{k}^{-1}(n)$ can be computed at the $k^{th}$ transmitter based on \emph{local} information only. {\em Theorem \ref{thm:optimal_soln}} also implies  that the smallest achievable value of $\beta_{k}\left(\mathbf{D}_{k}(n)\right)$ is the lower bound given by \eqref{eqn:lowerBD_G}. Based on that, we get an alternative sufficient condition for the iteration \eqref{eqn:system_iteration} to be a \emph{block-contraction mapping} w.r.t. the block-maximum norm for each subsystem $q \in \mathcal{Q}$, which is summarized into the following corollary.
\begin{Cor} [An Alternative Sufficient Condition] \label{cor:new_condition}
When the scaling matrices are chosen to be $\mathbf{D}_{k}(n) = -\partial_{kk}^{2}C_{k}(n), \forall~ k
\in \mathcal{K}$, an alternative sufficient condition for iteration \eqref{eqn:system_iteration} to be a \emph{block-contraction mapping} w.r.t. the vector block-maximum norm is given by
\begin{IEEEeqnarray}{rCl} \label{eqn:mse_prob_cor}
\max_{k \in \mathcal{K}} \left\{ \sum_{j =1, j \neq k}^{K} \max_{s \in \mathcal{S}}\frac{g_{kj}^{(s)}(q)}{g_{kk}^{(s)}(q)}\right\} < 1, \forall~ q \in \mathcal{Q},
\end{IEEEeqnarray}
where $g_{kj}^{(s)}(q) \triangleq g_{kj}^{(s)}(n)$ denotes the power gain when the FSMC $\{\mathbf{h}(n)\}$ is in state $q \in \mathcal{Q}$.
\end{Cor}
\begin{Rem}
The condition given in equation \eqref{eqn:mse_prob_cor} coincides with the condition given in Theorem 3 of \cite{Scutari2008_Unified} with the \emph{weight vector} chosen to be $\mathbf{1}$. Note that the condition given in \cite{Scutari2008_Unified} is for the \emph{iterative water-filling algorithm} (IWFA) \cite{Scutari2008II}, while {\em Corollary \ref{cor:new_condition}} here is concerned with the proposed DSGPA.
While the conclusion that the simultaneous DSGPA has similar convergence speed as the simultaneous IWFA given in \cite{Scutari2008II} is  based on numerical experiments, here we establish a theoretical foundation for that conclusion.
\end{Rem}

\textcolor{black}{\section{Numerical Results and Discussions}\label{sec:num_res_disc}}
In this section, we shall compare the proposed DSGPA with three baseline schemes: (I) Baseline $1$: gradient projection algorithm with a general positive definite scaling matrix (Gen-GPA); (II) Baseline $2$: gradient projection algorithm with a diagonal scaling matrix (Dia-GPA), whose diagonal entries are the diagonal entries of the corresponding Hessian matrix \cite{Xi2008, Stanczak2008, Bertsekas1989, Boyd2004}; (III) Baseline $3$: gradient projection algorithm with a constant stepsize (Con-GPA) $\xi = 0.005$ \cite{Bertsekas1989, Boyd2004, Scutari2008II, Zhang2008}. \textcolor{black}{We choose these three baselines as they have covered the majority of the existing approaches (baseline 2 and baseline 3) and also have an idea about what's the best possible performance (baseline 1). Specifically, baseline $1$ does not impose any {\em block diagonal} structure on the scaling matrices and requires centralized implementation. Hence, comparison with this baseline illustrates the potential performance loss of our proposed scheme due to the imposed block-diagonal structure in the scaling matrices. For baseline $2$ and baseline $3$, they are used  extensively  in \cite{Xi2008, Stanczak2008, Bertsekas1989, Boyd2004} and \cite{Bertsekas1989, Boyd2004, Scutari2008II, Zhang2008}, respectively.  Comparison with these baselines illustrates the performance improvement of the proposed scheme over these existing approaches. }

In all the simulations, there are $10$ randomly placed transmitter-receiver pairs, sharing $32$
independent subbands, i.e., it is chosen that $K = 10, N_{F}
= 32$. The total bandwidth is $10$ MHz. The maximum transmit power at the $k^{th}$ transmitter is set to be $1$ Watt, i.e., $P_{k, max} = 1, \forall k \in \mathcal{K}$. The distance from the $k^{th}$ transmitter
to the $j^{th} ~(\forall j \neq k)$ receiver is set to be $400$ meters, while the distance from the $k^{th}$ transmitter
to the $k^{th}$ receiver is set to be $100$ meters.  The path-loss exponent is $3.5$. The small scale fading channel gain is generated according to the distribution $\mathcal{CN}(0; 1)$. Moreover, we choose  $\widetilde{Q}_{kj}^{(s)} = 4$ for the
FSMC $\left\{h_{kj}^{(s)}(n)\right\}$, $\forall~ k, j \in
\mathcal{K}, s \in \mathcal{S}$ and the state space $\mathcal{\widetilde{H}}_{kj}^{(s)}$ of the FSMC $\left\{h_{kj}^{(s)}(n)\right\}$ is constructed based on the receiving SNR partition approach  \cite{Zhang2000, Wang1995}.

\textcolor{black}{\subsection{Tracking Performance Comparison}}\label{subsec:sim_tracking}
\textcolor{black}{Fig. \ref{fig:sim_1_tracking} illustrates the normalized sum-utility versus time-slot index for the proposed DSGPA and the three baseline schemes. As illustrated, the proposed DSGPA has a much better tracking capability than the baseline schemes Con-GPA and Dia-GPA, which are designed for quasi-static CSI. On the other hand, the DSGPA has similar performance as the centralized solution Gen-GPA, which shows that performance loss incurred by the block-diagonal structure of the scaling matrix used in the proposed DSGPA is negligible.}

\textcolor{black}{\subsection{Region Stability Property}}\label{subsec:sim_region}
\textcolor{black}{Fig. \ref{fig:sim_2_region} shows the simulation results of region stability. The simulation results are consistent with the analytical results stated in {\em Theorem \ref{thm:region_stability}}, i.e., the probability that the algorithm trajectory at steady state being out of the limit region $\mathcal{L}$ (see equation \eqref{eqn:limit_region}) is proportional to the normalized update interval $1/\overline{N}$.  Moreover, as the scaling matrices in the proposed DSGPA are adaptive to the time-varying CSI, the DSGPA performs better than the baseline schemes Con-GPA  and Dia-GPA . On the other hand, the performance of the DSGPA and the centralized  solution Gen-GPA  are similar.}

\textcolor{black}{\subsection{Order of Growth of the Tracking Errors}}\label{subsec:sim_errors}
\textcolor{black}{Fig. \ref{fig:sim_3_EAE} and Fig. \ref{fig:sim_4_MSE} show the simulation results of the expected-absolute-error (EAE) and the mean-square-error (MSE), respectively. Both figures are consistent with the analytical results given in {\em Theorem \ref{thm:UBD_MSE}}, i.e., the tracking errors, namely EAE and MSE, are proportional to the normalized update interval $1/\overline{N}$. Moreover, as the scaling matrices in the proposed DSGPA are adaptive to the time-varying CSI, the tracking errors associated with the  DSGPA are much smaller than the baseline schemes Con-GPA  and Dia-GPA . On the other hand, the performance  of the DSGPA and the centralized solution Gen-GPA  are quite similar.}

\section{Conclusions} \label{sec:conclusion}
In this paper, we have proposed a distributed scaled gradient projection algorithm (DSGPA) to solve the power control game in a $K$ pair multicarrier interference network under the finite-state Markov channel (FSMC) model.   We have shown that the proposed DSGPA converges to a limit region rather than a single point under the FSMC model. We have also shown that the order of growth of the tracking errors, namely the expected-absolute-error (EAE) and the mean-square-error (MSE), is given by $\mathcal{O}\left(1 \big/ \overline{N}\right)$.  By exploiting the Markovian property of the FSMC, the scaling matrix optimization problem (w.r.t. tracking error) is modeled as an infinite horizon average cost MDP. While there is no simple solution for MDP problems, we exploit the specific structure in the transition kernel and derive a low complexity distributive solution for controlling the scaling matrices to minimize the tracking errors. \textcolor{black}{Simulations are done to verify the analytical results as well as to demonstrate the superior performance of the proposed DSGPA over three baseline  schemes.}

\appendices
\textcolor{black}{\section{Proof of {\em Theorem \ref{thm:region_stability}}} \label{app:region_stability}
Consider a time interval $\left[0, ~N\right]$ with $M$ switchings, i.e.,  there are $M$ stages in the interval $\left[0,
~N\right]$. Let $\{q_{1}, q_{2}, \cdots, q_{M}\}$, $\{N_{1}, N_{2}, \cdots, N_{M}\}$ and $\{\phi_{1}, \phi_{2}, \cdots, \phi_{M}\}$ denote the channel states, sojourn times and contraction modulus of the $M$  stages, respectively, as shown in Fig. \ref{fig:stage}.  Under the conditions that $\big\{\beta_{k}\left(\mathbf{D}_{k}(n)\right) < 1,\forall~ k \in \mathcal{K}, \forall~ n \geq 1\big\}$, the iteration \eqref{eqn:system_iteration} is block-contraction in each stage \cite{Bertsekas1989, Scutari2008II}. As a result, the distances $\left\{ \left\|\mathbf{e}(m) \right\|_{\textmd{block}}, 1 \leq m \leq M \right\}$  between the iterate $\mathbf{p}(n)$ and the NE at the end of each stage can be upper bounded by
\begin{IEEEeqnarray}{rCl} \label{eqn:geo}
\left\|\mathbf{e}(1) \right\|_{\textmd{block} } &\leq& \mathbf{p}(0)\phi_{1}^{N_{1}}, \\
\left\|\mathbf{e}(m) \right\|_{\textmd{block} } &\leq& \left\|\mathbf{e}(m-1) \right\|_{\textmd{block} }\phi_{m}^{N_{m}} + \delta_{m-1, m}, ~2 \leq m \leq M, \label{eqn:em}
\end{IEEEeqnarray}
where $\delta_{m-1, m}$ denotes the distance between the NE of the $(m-1)^{th}$ stage and $m^{th}$ stage, i.e., the jump of the equilibrium point of the switched system. Iterating  equation \eqref{eqn:em} from $m = 1$ to $m = M$, we get
\begin{IEEEeqnarray}{rCl} \label{eqn:WNt}
\left\|\mathbf{e}(M) \right\|_{\textmd{block} } &\leq& \mathbf{p}(0) \prod_{m=1}^{M} \phi_{m}^{N_{m}} + \sum_{l = 2}^{M} \prod_{m=l}^{M} \delta_{m-1, m}\phi_{m}^{N_{m}} \notag \\
&\leq& \mathbf{p}(0) \beta^{\sum_{m=1}^{M}N_{m}} + \delta\sum_{l = 2}^{M} \prod_{m=l}^{M} \beta^{N_{m}}, \label{eqn:sumprod}
\end{IEEEeqnarray}
where $\beta = \max_{m}\phi_{m} = \max_{\{k,n\}}\beta_{k}\left(\mathbf{D}_{k}(n)\right)$ is the worst case contraction modulus.}

Let $\mathbf{q} = [q_{1} \:\ q_{2} \:\
\cdots \:\ q_{M}]$ and $\mathbf{N} = [N_{1} \:\ N_{2} \:\ \cdots \:\ N_{M}]$, then take expectation of both sides of equation \eqref{eqn:sumprod} w.r.t. $\left\{\mathbf{q}, \mathbf{N} \right\}$, we can get
\begin{IEEEeqnarray}{rCl} \label{eqn:iterating}
\mathbb{E}_{ \left\{\mathbf{q}, \mathbf{N} \right\} }\left\{ \left\|\mathbf{e}(M) \right\|_{\textmd{block} } \right\} \leq \mathbb{E}_{ \left\{\mathbf{q}, \mathbf{N} \right\} }\left\{\mathbf{p}(0) \beta^{\sum_{m=1}^{M}N_{m}} \right\} + \mathbb{E}_{ \left\{\mathbf{q}, \mathbf{N} \right\} }\left\{\delta\sum_{l = 2}^{M} \prod_{m=l}^{M} \beta^{N_{m}} \right\}.
\end{IEEEeqnarray}

Under the assumption that $\left[\mathbf{T}\right]_{qq} = \nu^{K^{2}N_{F}}, \forall~ q \in \mathcal{Q}$ (see equation \eqref{eqn:TPM}), then we know that $N_{1}, N_{2}, \cdots,$ $N_{M}$ are identically distributed with probability mass function (PMF) given by
\begin{IEEEeqnarray}{rCl} \label{eqn:pmf}
Pr\left\{ N_{m} = l \right\} = \nu^{K^{2}N_{F}(l-1)}\left(1-\nu^{K^{2}N_{F}}\right), ~\forall~ l = 1, 2, \cdots .
\end{IEEEeqnarray}

Then, we have
\begin{IEEEeqnarray}{rCl} \label{eqn:expectation}
\mathbb{E}_{ \left\{\mathbf{q}, \mathbf{N} \right\} }\left\{\delta\sum_{l = 2}^{M} \prod_{m=l}^{M} \beta^{N_{m}} \right\} = \mathbb{E}_{ \left\{\mathbf{q}\right\} }\left\{   \mathbb{E}_{ \left\{\mathbf{N}\right\} } \left\{  \delta\sum_{l = 2}^{M} \prod_{m=l}^{M} \beta^{N_{m}} \bigg| \mathbf{q} \right\}  \right\} = \frac{\delta\alpha\left(1-\alpha^{M-1}\right)}{1-\alpha},
\end{IEEEeqnarray}
where $\alpha = \mathbb{E}\left\{ \beta^{N_{m}} \right\} = \frac{ \beta \left(1-\nu^{K^{2}N_{F}}\right) }{ 1 - \beta \nu^{K^{2}N_{F}} } = \frac{\beta}{\beta + (1-\beta)\overline{N}}$. Therefore, let $M \rightarrow + \infty$, we get
\begin{IEEEeqnarray}{rCl} \label{eqn:Minfty}
\mathbb{E}_{ \left\{\mathbf{q}, \mathbf{N} \right\} }\left\{ \left\|\mathbf{e}(\infty) \right\|_{\textmd{block}} \right\} \leq \frac{\delta \alpha}{1-\alpha} = \frac{\delta\beta}{(1-\beta)\overline{N}}.
\end{IEEEeqnarray}

Then, by virtue of the Markov inequality  we get
\begin{IEEEeqnarray}{rCl} \label{eqn:MarkovIe}
Pr\left\{ \left\|\mathbf{e}(\infty) \right\|_{\textmd{block} }  > \delta\right\} \leq
\frac{\beta}{(1-\beta)\overline{N}},
\end{IEEEeqnarray}
which means
\begin{IEEEeqnarray}{rCl} \label{eqn:region_prob_app}
\lim_{n \rightarrow +\infty }Pr\left\{\mathbf{p}(n) \notin \mathcal{L}\right\} \leq \min_{}\left\{1, \frac{\beta}{(1-\beta)\overline{N}}\right\}.
\end{IEEEeqnarray}

The the sojourn time of  the FSMC $\left\{\mathbf{h}(n)\right\}$ is geometrically distributed with parameter $\nu^{K^{2}N_{F}}$, and the \emph{average sojourn time} of  the FSMC $\left\{\mathbf{h}(n)\right\}$ is given by $\frac{1}{1-\nu^{K^{2}N_{F}}}$.

\section{ Proof of { \em Theorem \ref{thm:UBD_MSE} } } \label{app:UBD_MSE}
From equation \eqref{eqn:Minfty} in {\em Appendix  \ref{app:region_stability}} we know that, at steady state (i.e., when $n \rightarrow + \infty$),
\begin{IEEEeqnarray}{rCl} \label{eqn:steady_state_EAE}
\mbox{EAE}\left(\mathbf{p}(n)\right) &\leq&  \frac{ \delta\beta}{(1-\beta)\overline{N}}.
\end{IEEEeqnarray}

Let $\varpi =  \mathbb{E}\left\{ \beta^{2N_{m}} \right\} = \frac{ \beta^{2} \left(1-\nu^{K^{2}N_{F}}\right) }{ 1 - \beta^{2} \nu^{K^{2}N_{F}} } = \frac{\beta^{2}}{\beta^{2} + (1-\beta^{2})\overline{N}}$,  we then have
\begin{IEEEeqnarray}{rCl} \label{eqn:MSE_sq}
&&\mathbb{E}_{ \left\{\mathbf{q}, \mathbf{N} \right\} }\left\{ \left( \delta\sum_{l = 2}^{M} \prod_{m=l}^{M} \beta^{N_{m}} \right)^2 \right\} = \mathbb{E}_{ \left\{\mathbf{q}\right\} }\left\{   \mathbb{E}_{ \left\{\mathbf{N}\right\} } \left\{  \left( \delta\sum_{l = 2}^{M} \prod_{m=l}^{M} \beta^{N_{m}} \right)^2\bigg| \mathbf{q} \right\}  \right\} \notag \\
&&= \frac{2\varpi}{1-\varpi}\left[\frac{1-\alpha^{M-1}}{1-\alpha}  - \frac{\varpi^{M} - \varpi\alpha^{M-1}}{\varpi-\alpha}\right] - \frac{\varpi(1-\varpi^{M-1})}{1-\varpi}.
\end{IEEEeqnarray}

Combining equation \eqref{eqn:iterating} and equation \eqref{eqn:MSE_sq}, and let $M \rightarrow + \infty$, we can get
\begin{IEEEeqnarray}{rCl} \label{eqn:app_UBD_MSE}
\mbox{MSE}\left(\mathbf{p}(n)\right) &\leq& \frac{ \delta^{2}\beta^{2}\left(2\beta + (1-\beta) \overline{N}\right) }{(1-\beta^{2})(1-\beta)\overline{N}^{2}}.
\end{IEEEeqnarray}

\section{ Proof of {\em Lemma \ref{lem:dominanting_error}} } \label{app:dominanting_error}
Under the conditions that $\big\{\beta_{k}\left(\mathbf{D}_{k}(n)\right) < 1,\forall~ k \in \mathcal{K}, \forall~ n \geq 1\big\}$, the iteration \eqref{eqn:system_iteration} is block-contraction in each stage \cite{Bertsekas1989, Scutari2008II}. From $\widetilde{e}(1) = e(1)$, we can get
\begin{IEEEeqnarray}{rCl} \label{eqn:e1}
e(1+) \leq e(1)\beta_{1}^{N_{1}} = \widetilde{e}(1)\beta_{1}^{N_{1}},
\end{IEEEeqnarray}
where $e(1+)$ denotes the distance between the algorithm trajectory and the NE at the end of  stage 1. Since $e(2) \leq e(1+) + \delta_{1, 2}$, we then have
\begin{IEEEeqnarray}{rCl} \label{eqn:e2}
e(2) \leq e(1)\beta_{1}^{N_{1}} + \delta_{1, 2} = \widetilde{e}(1)\beta_{1}^{N_{1}} + \delta_{1, 2} = \widetilde{e}(2).
\end{IEEEeqnarray}

Repeat the same procedure for each stage (i.e., $m = 3, 4, \cdots$), we can get
\begin{IEEEeqnarray}{rCl} \label{eqn:error_em}
e(m) \leq \widetilde{e}(m), \forall~ m = 1, 2, \cdots .
\end{IEEEeqnarray}

Moreover, set $N_{m} = 1, \forall~ m$, in equation \eqref{eqn:WNt} and let $M \rightarrow + \infty$, we can get $\widetilde{e}(m) \leq \frac{\delta\beta}{1-\beta}$. Since the above arguments hold for all realizations of the FSMC $\{\mathbf{h}(n)\}$, equation \eqref{eqn:almost_sure} holds almost surely.

\section{ Proof of {\em Theorem \ref{thm:solution_mdp}} } \label{app:solution_mdp}
Define $P\left( \bm{\chi}_{lr} \big| \bm{\chi}_{iq}, \pi\left( \bm{\chi}_{iq} \right) \right) \triangleq Pr\left\{ \bm{\chi}(n+1) = (\bar{e}_{l}, \mathbf{h}_{r}) \big| \bm{\chi}(n) = (\bar{e}_{i}, \mathbf{h}_{q}),   \pi\left( \bm{\chi}(n)\right)  \right\}$, the optimal control policy $\pi^{*}$ can be found by solving the Bellman equation below\cite{Bertsekas2005}:
\begin{IEEEeqnarray}{rCl} \label{eqn:Bellman}
V\left( \bm{\chi}_{iq} \right) + \theta = \minimize{\pi\left( \bm{\chi}_{iq} \right)}\left\{  g\left(\bar{e}_{i}\right) + \sum_{\bm{\chi}_{lr}} P\left( \bm{\chi}_{lr} \big| \bm{\chi}_{iq}, \pi\left( \bm{\chi}_{iq} \right) \right) V\left( \bm{\chi}_{lr} \right) \right\},
\end{IEEEeqnarray}
\textcolor{black}{where $V\left( \bm{\chi}_{iq} \right)$ is called the optimal value function at state $\bm{\chi}_{iq}$; and  $\theta$ is the optimal average cost. The Bellman equation is a fixed point equation w.r.t. the variables $\left\{V\left( \bm{\chi}_{iq} \right), 1\leq i,q \leq Q\right\}$ and $\theta$, which is derived using the principle of {\em divide-and-conquer} from the original MDP in \eqref{eqn:MDPformulation}. It is shown using the theory of  MDP \cite{Bertsekas2005} that the fixed point solution $\theta$ of \eqref{eqn:Bellman} would give the optimal value of $J^{\pi}$ in \eqref{eqn:MDPformulation} and the optimal control policy $\pi^{*}$ is given by the solution of the RHS of \eqref{eqn:Bellman} w.r.t. the fixed point solution $\left\{V\left( \bm{\chi}_{iq} \right), 1\leq i,q \leq Q\right\}$. For the transition kernel $P\left( \bm{\chi}_{lr} \big| \bm{\chi}_{iq}, \pi\left( \bm{\chi}_{iq} \right) \right)$ given in \eqref{eqn:MDPkernel}, we have  $\frac{ \bar{e}_{l} - \delta_{qr} }{ \bar{e}_{i} } < 1$, since $\bar{e}_{l} \leq \bar{e}_{i} [\beta(n)]^{N_{q}} + \delta_{qr} < \bar{e}_{i} + \delta_{qr}$, where $N_{q}$ denotes the sojourn time of state $q$ of the FSMC $\{\mathbf{h}(n)\}$. As a result, minimizing $\beta(n)$ can simultaneously minimize all the transition probabilities $P\left( \bm{\chi}_{lr} \big| \bm{\chi}_{iq}, \pi\left( \bm{\chi}_{iq} \right) \right), \forall~ 1 \leq l \leq L, 1 \leq r \leq Q$, which suggests that the optimal control policy $\pi^{*}$ is to take the minimum possible $\beta(n)$ at state $\bm{\chi}_{iq}, \forall~ 1 \leq i \leq L, 1 \leq q \leq Q$. Moreover, since $\beta(n) = \max_{\{\forall~k \in \mathcal{K}\}} \left\{\beta_{k}\left(\mathbf{D}_{k}(n)\right) \right\}$, the minimum possible $\beta(n)$ is given by solving the optimization problem given in equation \eqref{eqn:mse_prob} at each time-slot.}

\section{ Proof of {\em Lemma \ref{lem:Trans_sumprob}} } \label{app:Trans_sumprob}
To prove {\em Lemma \ref{lem:Trans_sumprob}}, we need the the following intermediate results.
\begin{Lem} [Matrix Induced-2 Norm of Product of Two Matrices] \label{lem:eigenvalue_AB}
     For two symmetric positive definite matrices $\mathbf{A}, \mathbf{B} \in \mathbb{C}^{m \times m}$, we have
     \begin{IEEEeqnarray}{rCl} \label{eqn:eigenvalue_lemma}
     \left\|\mathbf{AB}\right\|_{2} \geq \lambda_{min}\left(\mathbf{A}\right)\left\|\mathbf{B}\right\|_{2},
     \end{IEEEeqnarray}
     where $\lambda_{min}\left(\mathbf{A}\right)$ denotes the minimum eigenvalue of matrix $\mathbf{A}$.
     \end{Lem}
     \begin{Proof}
Given that $\mathbf{A}$ is positive definite, it then follows that $\mathbf{A^{2}} - \lambda_{min}^{2}\left(\mathbf{A}\right)\mathbf{I}_{m}$ is positive semidefinite. As $\mathbf{B}$ is also positive definite, then $\left(\mathbf{A^{2}} - \lambda_{min}^{2}\left(\mathbf{A}\right)\mathbf{I}_{m}\right)\mathbf{B}^{2}$ is positive semidefinite \cite{Horn1985}. Denoting $\mathbf{a} \in \mathbb{C}^{m}$ and $\mathbf{b} \in \mathbb{C}^{m}$ as the \emph{unit norm}  eigenvectors corresponding to the \emph{largest} eigenvalues of matrices $\mathbf{A}^{2}\mathbf{B}^{2}$ and $\mathbf{B}^{2}$, respectively, we have
\begin{IEEEeqnarray}{rCl} \label{eqn:aAbB}
\left\|\mathbf{AB}\right\|_{2}^{2} =  \mathbf{a}^{\dag}\mathbf{A}^{2}\mathbf{B}^{2}\mathbf{a} \geq \mathbf{b}^{\dag}\mathbf{A}^{2}\mathbf{B}^{2}\mathbf{b} \geq \lambda_{min}^{2}\left(\mathbf{A}\right)\mathbf{b}^{\dag}\mathbf{B}^{2}\mathbf{b}  = \lambda_{min}^{2}\left(\mathbf{A}\right)\left\|\mathbf{B}\right\|_{2}^{2}.
\end{IEEEeqnarray}
Therefore, from the definition of matrix-2 norm \cite{Horn1985}, we get $\left\|\mathbf{AB}\right\|_{2} \geq \lambda_{min}\left(\mathbf{A}\right)\left\|\mathbf{B}\right\|_{2}$.
\end{Proof}

We next proceed to prove {\em Lemma \ref{lem:Trans_sumprob}}.
Define $\mathbf{\tilde{F}}(n) = -diag\left(\left[ \frac{ g_{kj}^{(1)}(n) }{ g_{kk}^{(1)}(n) } \:\ \frac{ g_{kj}^{(2)}(n) }{ g_{kk}^{(2)}(n)} \:\ \cdots \:\ \frac{ g_{kj}^{(N_{F})}(n) }{ g_{kk}^{(N_{F})}(n) } \right] \right)$ and $\mathbf{\widetilde{D}}_{k}(n) = -\mathbf{D}_{k}^{-1}(n) \partial_{kk}^{2}C_{k}(n) $, we then have
\begin{IEEEeqnarray}{rCl} \label{eqn:gDi}
\beta_{k}\left(\mathbf{D}_{k}(n)\right) =      \left\|\mathbf{I}_{N_{F}} -  \mathbf{\widetilde{D}}_{k}(n) \right\|_{2} + \sum_{j =1, j \neq k}^{K} \left\| \mathbf{\widetilde{D}}_{k}(n)\mathbf{\widetilde{F}}(n)    \right\|_{2}.
\end{IEEEeqnarray}

Moreover, from {\em Lemma \ref{lem:eigenvalue_AB} } we know that
\begin{IEEEeqnarray}{rCl} \label{eqn:gDiLower}
\beta_{k}\left(\mathbf{D}_{k}(n)\right) &\geq& \left|1 - \lambda_{min}\left(\mathbf{\widetilde{D}}_{k}(n)\right)\right| + \lambda_{min}\left(\mathbf{\widetilde{D}}_{k}(n)\right) \sum_{j =1, j \neq k}^{K} \left\| \mathbf{\widetilde{F}}(n)    \right\|_{2} \label{eqn:bound1} \\
&\stackrel{(a)}{\geq}& 1 + \lambda_{min}\left(\mathbf{\widetilde{D}}_{k}(n)\right) \left( \sum_{j =1, j \neq k}^{K} \left\| \mathbf{\widetilde{F}}(n)    \right\|_{2} - 1 \right) \stackrel{(b)}{\geq} \sum_{j =1, j \neq k}^{K} \left\| \mathbf{\widetilde{F}}(n)    \right\|_{2}.
\end{IEEEeqnarray}
where (a) is because  $\lambda_{min}\left(\mathbf{\widetilde{D}}_{k}\right) \leq 1$ achieves smaller $\beta_{k}\left(\mathbf{D}_{k}(n)\right)$ than $\lambda_{min}\left(\mathbf{\widetilde{D}}_{k}\right) > 1$; and  (b) is because $\beta_{k}\left(\mathbf{D}_{k}(n)\right) < 1$ and  choosing $\lambda_{min}\left(\mathbf{\widetilde{D}}_{k}\right) = 1$ is the best possible choice.

\section{ Proof of {\em Theorem \ref{thm:optimal_soln}} } \label{app:optimal_soln}
Similar to equation \eqref{eqn:gDiLower}, we can also get
\begin{IEEEeqnarray}{rCl} \label{eqn:optSOLN}
\beta\left(\mathbf{D}_{k}(n)\right) &\geq& \left| \lambda_{max}\left(\mathbf{\widetilde{D}}_{k}(n)\right) - 1 \right| + \lambda_{min}\left(\mathbf{\widetilde{D}}_{k}(n)\right) \sum_{j =1, j \neq k}^{K} \left\| \mathbf{\widetilde{F}}(n)    \right\|_{2}.
\end{IEEEeqnarray}

Equation \eqref{eqn:bound1} together with equation \eqref{eqn:optSOLN} imply that we shall choose $\lambda_{max}\left(\mathbf{\widetilde{D}}_{k}(n)\right) = \lambda_{min}\left(\mathbf{\widetilde{D}}_{k}(n)\right) = 1$, which achieves the lower bound of $\beta\left(\mathbf{D}_{k}(n)\right)$. Therefore, one of the optimal solutions of {\em Problem \ref{prob:mse_submini}} is given by $\mathbf{\widetilde{D}}_{k}(n) = \mathbf{I}_{N_{F}}$, i.e., $\mathbf{D}_{k}(n) = -\partial_{kk}^{2}C_{k}(n)$.

\bibliographystyle{IEEEtran}
\bibliography{IEEEabrv,references_cited}

\begin{thebibliography}{10}
\providecommand{\url}[1]{#1}
\csname url@samestyle\endcsname
\providecommand{\newblock}{\relax}
\providecommand{\bibinfo}[2]{#2}
\providecommand{\BIBentrySTDinterwordspacing}{\spaceskip=0pt\relax}
\providecommand{\BIBentryALTinterwordstretchfactor}{4}
\providecommand{\BIBentryALTinterwordspacing}{\spaceskip=\fontdimen2\font plus
\BIBentryALTinterwordstretchfactor\fontdimen3\font minus
  \fontdimen4\font\relax}
\providecommand{\BIBforeignlanguage}[2]{{%
\expandafter\ifx\csname l@#1\endcsname\relax
\typeout{** WARNING: IEEEtran.bst: No hyphenation pattern has been}%
\typeout{** loaded for the language `#1'. Using the pattern for}%
\typeout{** the default language instead.}%
\else
\language=\csname l@#1\endcsname
\fi
#2}}
\providecommand{\BIBdecl}{\relax}
\BIBdecl

\bibitem{Yu2002}
W.~Yu, G.~Ginis, and J.~M. Cioffi, ``Distributed multiuser power control for
  digital subscriber lines,'' \emph{{IEEE} J. Sel. Areas Commun.}, vol.~20,
  no.~5, pp. 1105--115, 2002.

\bibitem{Scutari2008_I}
G.~Scutari, D.~Palomar, and S.~Barbarossa, ``Optimal linear precoding
  strategies for wideband noncooperative systems based on game theory -- part
  {I}: Nash equilibria,'' \emph{{IEEE} Trans. Signal Process.}, vol.~56, no.~3,
  pp. 1230--1249, 2008.

\bibitem{Scutari2008II}
------, ``Optimal linear precoding strategies for wideband non-cooperative
  systems based on game theory -- part {II}: Algorithms,'' \emph{{IEEE} Trans.
  Signal Process.}, vol.~56, no.~3, pp. 1250--1267, 2008.

\bibitem{Palomar2006}
D.~Palomar and M.~Chiang, ``A tutorial on decomposition methods for network
  utility maximization,'' \emph{{IEEE} J. Sel. Areas Commun.}, vol.~24, no.~8,
  pp. 1439--1451, Aug. 2006.

\bibitem{Chiang2005}
M.~Chiang, ``Balancing transport and physical layers in wireless multihop
  networks: Jointly optimal congestion control and power control,''
  \emph{{IEEE} Trans. Inf. Theory}, vol.~23, no.~1, pp. 104--116, Jan. 2005.

\bibitem{Huang2005}
J.~Huang, R.~Berry, and M.~Honig, ``Performance of distributed utility-based
  power control for wireless ad hoc networks,'' in \emph{IEEE MILCOM'05.},
  vol.~4, Oct. 2005, pp. 2481--2487.

\bibitem{Huang2006}
------, ``Distributed interference compensation for wireless networks,''
  \emph{{IEEE} J. Sel. Areas Commun.}, vol.~24, no.~5, pp. 1074--1084, May
  2006.

\bibitem{Stanczak2007}
S.~Stanczak, M.~Wiczanowski, and H.~Boche, ``Distributed utility-based power
  control: Objectives and algorithms,'' \emph{{IEEE} Trans. Signal Process.},
  vol.~55, no.~10, pp. 5058--5068, Oct. 2007.

\bibitem{Stanczak2008}
------, \emph{Fundamentals of Resource Allocation in Wireless Networks: Theory
  and Algorithms}.\hskip 1em plus 0.5em minus 0.4em\relax 2nd Ed., Springer,
  2008.

\bibitem{Boyd2006}
S.~Boyd, A.~Ghosh, B.~Prabhakar, and D.~Shah, ``Randomized gossip algorithms,''
  \emph{{IEEE} Trans. Inf. Theory}, vol.~52, no.~6, pp. 2508--2530, Jun. 2006.

\bibitem{Alpcan2004}
T.~Alpcan and T.~Basar, ``A hybrid systems model for power control in multicell
  wireless data networks,'' \emph{Performance Evaluation}, vol.~57, no.~4, pp.
  477--495, 2004.

\bibitem{Podelski_Proof}
A.~Podelski and S.~Wagner, \emph{Region Stability Proofs for Hybrid
  Systems}.\hskip 1em plus 0.5em minus 0.4em\relax Formal Modeling and Analysis
  of Timed Systems. Springer Berlin, 2007.

\bibitem{Costa2007}
A.~Costa and F.~J. V\'{a}zquez-Abad, ``Adaptive stepsize selection for tracking
  in a regime-switching environment,'' \emph{Automatica}, vol.~43, no.~11, pp.
  1896--1908, 2007.

\bibitem{Yin2005}
G.~Yin and V.~Krishnamurthy, ``Least mean square algorithms with markov
  regime-switching limit,'' \emph{{IEEE} Trans. Autom. Control}, vol.~50,
  no.~5, pp. 577--593, May 2005.

\bibitem{Berenguer2005}
I.~Berenguer, X.~Wang, and V.~Krishnamurthy, ``Adaptive {MIMO} antenna
  selection via discrete stochastic optimization,'' \emph{{IEEE} Trans. Signal
  Process.}, vol.~53, no.~11, pp. 4315--4329, Nov. 2005.

\bibitem{Krishnamurthy2005}
V.~Krishnamurthy, C.~Athaudage, and D.~Huang, ``Adaptive {OFDM} synchronization
  algorithms based on discrete stochastic approximation,'' \emph{{IEEE} Trans.
  Signal Process.}, vol.~53, no.~4, pp. 1561--1574, Apr. 2005.

\bibitem{Liberzon2003}
D.~Liberzon, \emph{Switching in Systems and Control}.\hskip 1em plus 0.5em
  minus 0.4em\relax Birkhauser, Jun., 2003.

\bibitem{Huang2009}
K.~Huang, R.~W. {Heath Jr.}, and J.~G. Andrews, ``Limited feedback beamforming
  over temporally correlated channels,'' \emph{{IEEE} Trans. Signal Process.},
  vol.~57, no.~5, pp. 1959--1975, May 2009.

\bibitem{Zhang2009}
\BIBentryALTinterwordspacing
J.~Zhang, R.~W. {Heath Jr.}, M.~Kountouris, and J.~G. Andrews, ``Multi-mode
  transmission for the {MIMO} broadcast channel with imperfect channel state
  information,'' \emph{submitted to {IEEE} Trans. Wireless Commun.}, 2009.
  [Online]. Available:
  \url{http://arxiv.org/PS\_cache/arxiv/pdf/0903/0903.5108v1.pdf}
\BIBentrySTDinterwordspacing

\bibitem{YIN2008}
G.~G. Yin, C.-A. Tan, L.~Y. Wang, and C.~Xu, ``Recursive estimation algorithms
  for power controls of wireless communication networks,'' \emph{Journal of
  Control Theory and Applications}, vol.~6, no.~3, pp. 225--232, 2008.

\bibitem{Babich2000}
F.~Babich and G.~Lombardi, ``A markov model for the mobile propagation
  channel,'' \emph{{IEEE} Trans. Veh. Technol.}, vol.~49, no.~1, pp. 63--73,
  Jan. 2000.

\bibitem{Zhang2000}
Q.~Zhang and S.~Kassam, ``Finite-state markov model for rayleigh fading
  channels,'' \emph{{IEEE} Trans. Commun.}, vol.~47, no.~11, pp. 1688--1692,
  Nov. 1999.

\bibitem{Wang1995}
H.~S. Wang and N.~Moayeri, ``Finite-state markov channel-a useful model for
  radio communication channels,'' \emph{{IEEE} Trans. Veh. Technol.}, vol.~44,
  no.~1, pp. 163--171, Feb 1995.

\bibitem{Diaconis1997}
\BIBentryALTinterwordspacing
P.~Diaconis and D.~Freedman, ``On markov chains with continuous state space,''
  Dec. 1997. [Online]. Available:
  \url{http://www.stat.berkeley.edu/tech-reports/501.pdf}
\BIBentrySTDinterwordspacing

\bibitem{Wilkinson2006}
D.~J. Wilkinson, \emph{Stochastic modelling for systems biology}.\hskip 1em
  plus 0.5em minus 0.4em\relax Chapman \& Hall/CRC, 2006.

\bibitem{Cheng2010}
Y.~Cheng and V.~K.~N. Lau, ``Iterative primal-dual scaled gradient algorithm
  with dynamic scaling matrices for solving distributive {NUM} over
  time-varying fading channels,'' \emph{submitted to {IEEE} Trans. Wireless
  Commun.}, Feb. 2010.

\bibitem{Bertsekas1989}
D.~P. Bertsekas and J.~N. Tsitsiklis, \emph{Parallel and Distributed
  Computation: Numerical Methods}.\hskip 1em plus 0.5em minus 0.4em\relax
  Athena Scientific, 1989.

\bibitem{Xi2008}
Y.~Xi and E.~Yeh, ``Node-based optimal power control, routing, and congestion
  control in wireless networks,'' \emph{{IEEE} Trans. Inf. Theory}, vol.~54,
  no.~9, pp. 4081--4106, Sept. 2008.

\bibitem{Boyd2004}
S.~Boyd and L.~Vandenberghe, \emph{Convex Optimization}.\hskip 1em plus 0.5em
  minus 0.4em\relax Cambridge University Press, 2004.

\bibitem{Arrow1958}
K.~Arrow, L.~Hurwicz, and H.~Uzawa, \emph{Studies in Linear and Non-Linear
  Programming}.\hskip 1em plus 0.5em minus 0.4em\relax Stanford University
  Press, Stanford, 1958.

\bibitem{Scutari2008_Unified}
G.~Scutari, D.~Palomar, and S.~Barbarossa, ``Competitive design of multiuser
  mimo systems based on game theory: A unified view,'' \emph{{IEEE} J. Sel.
  Areas Commun.}, vol.~26, no.~7, pp. 1089--1103, 2008.

\bibitem{Zhang2008}
J.~Zhang, D.~Zheng, and M.~Chiang, ``The impact of stochastic noisy feedback on
  distributed network utility maximization,'' \emph{{IEEE} Trans. Inf. Theory},
  vol.~54, no.~2, pp. 645--665, Feb. 2008.

\bibitem{Bertsekas2005}
D.~P. Bertsekas, \emph{Dynamic programming and optimal control}.\hskip 1em plus
  0.5em minus 0.4em\relax Athena Scientific, 2007.

\bibitem{Horn1985}
R.~Horn and C.~Johnson, \emph{Matrix Analysis}.\hskip 1em plus 0.5em minus
  0.4em\relax Cambridge University Press, 1985.

\end{thebibliography}

\begin{table}[h]
\centering \caption{\textcolor{black}{Summary of Main Notations}}\label{table:notations}
\begin{tabular} {|l|l|}
\hline \textbf{Notation} & \textbf{Meaning} \\
\hline $n$ & an index variable to denote the time-slot\\
\hline $m$ & an index variable to denote the stage\\
\hline $h_{kj}^{(s)}(n)$ & channel gain between the $k^{th}$
receiver and the $j^{th}$ transmitter on  the $s^{th}$ subband \\
\hline $\mathbf{h}(n)$ & collection of the channel coefficients of the entire network \\
\hline $\mathbf{D}_{k}(n)$  & scaling matrix of the $k^{th}$ transmitter at time-slot $n$\\
\hline $\mathbf{D}(n)$  & block-diagonal matrix, which consists of the $K$ scaling matrices $\left\{\mathbf{D}_{k}(n), \forall k \in \mathcal{K}\right\}$ \\
\hline $\beta_{k}\left(\mathbf{D}_{k}(n)\right)$ & contraction modulus of the $k^{th}$ transmitter at time-slot $n$  \\
\hline $\beta$ & worse case contraction modulus of all transmitters in all channel states (i.e., all time-slots) \\
\hline $\phi_{m}$ & worse case contraction modulus of all transmitters in the $m^{th}$stage\\
\hline $N_{q}$ & random variable, which is the sojourn time of channel state $q$ \\
\hline $N_{m}$ & random variable, which is the sojourn time of stage $m$ \\
\hline $\bar{\mathbf{p}}^{(q)}$  & NE of Game $\mathcal{G}$ at channel state $q$ \\
\hline $\delta$  & maximum distance between two NEs corresponding to two different channel states \\
\hline $\mathcal{L}$ & limit region, which is a polyhedron with diameter $\delta$ \\
\hline $\left\{\widetilde{e}(m)\right\}$ & dominated error process\\
\hline $\delta_{m, m+1}$ & distance between the two NEs corresponding to the $m^{th}$ stage and the $(m+1)^{th}$ stage  \\
\hline $\pi$ & scaling matrix control policy\\
\hline $\bm{\chi}(n)$ & system state at the $n^{th}$ time-slot\\
\hline $\left\{\bm{\chi}(n)\right\}$ & Markov chain induced by the control policy $\pi$\\
\hline
\end{tabular}
\end{table}

\begin{figure} [h]
\centering
\includegraphics[scale=0.7]{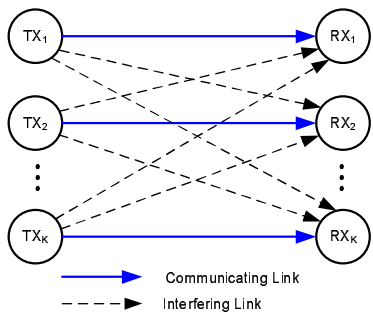}
\caption {A sample multicarrier interference network with $K$ transmitter-receiver
pairs, where the $k^{th}$ transmitter wishes to communicate with the $k^{th}$ receiver, $\forall~ k \in \mathcal{K}$. All the $K$ transmitters share $N_{F}$ nonoverlapping subcarriers.} \label{fig:system_model}
\end{figure}

\begin{figure} [h]
\centering
\includegraphics[scale=0.7]{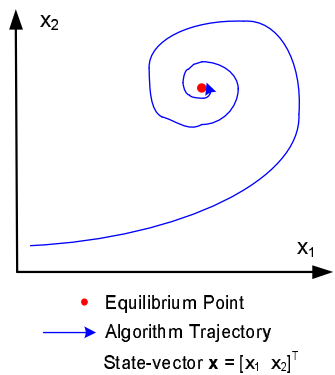}
\caption {A pictorial view of an algorithm trajectory converging to the static  Nash Equilibrium  point in a two-dimensional algorithm space.} \label{fig:trajectory}
\end{figure}

\begin{figure} [h]
\centering
\includegraphics[scale=0.6]{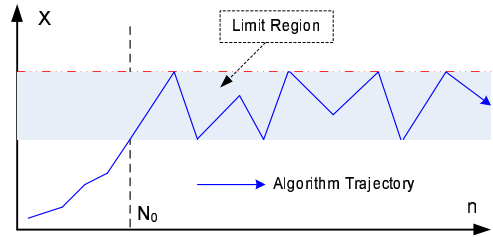}
\caption {An illustration of the region stability of a randomly switched system with one-dimensional state space \cite{Podelski_Proof}. From time-slot $N_{0}$ and onwards, the trajectory remains in the limit region with probability $P_{Region}$ given in equation \eqref{eqn:region_prob}.} \label{fig:region_stability}
\end{figure}

\begin{figure} [h]
\centering
\includegraphics[scale=0.65]{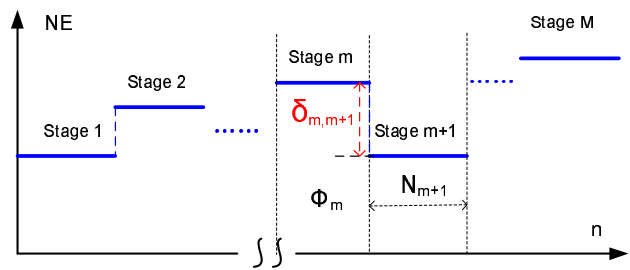}
\caption {An illustration of the $M$ stages that the switched system \eqref{eqn:system_iteration} has gone through. $\phi_{m}$ and $N_{m}$ denote the contraction modulus and the sojourn time of the $m^{th}$ stage, respectively; and  $\delta_{m, m+1}$ denotes the distance between the NE of the $m^{th}$ stage and the  $(m+1)^{th}$ stage, $\forall~ m = 1, 2, \cdots .$} \label{fig:stage}
\end{figure}

\newpage
\begin{figure} [h]
\centering
\includegraphics[scale=0.65]{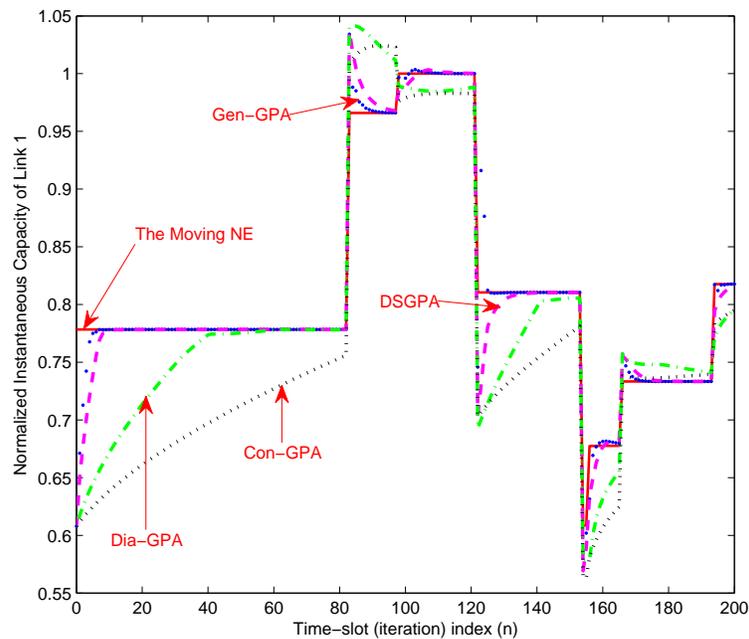}
\caption {\textcolor{black}{Tracking performance comparison of the proposed DSGPA and the baseline schemes: Gen-GPA, Dia-GPA and Con-GPA. The link capacity is normalized to the maximum link capacity obtained across all the time-slots.
}}\label{fig:sim_1_tracking}
\end{figure}

\newpage
\begin{figure} [h]
\centering
\includegraphics[scale=0.75]{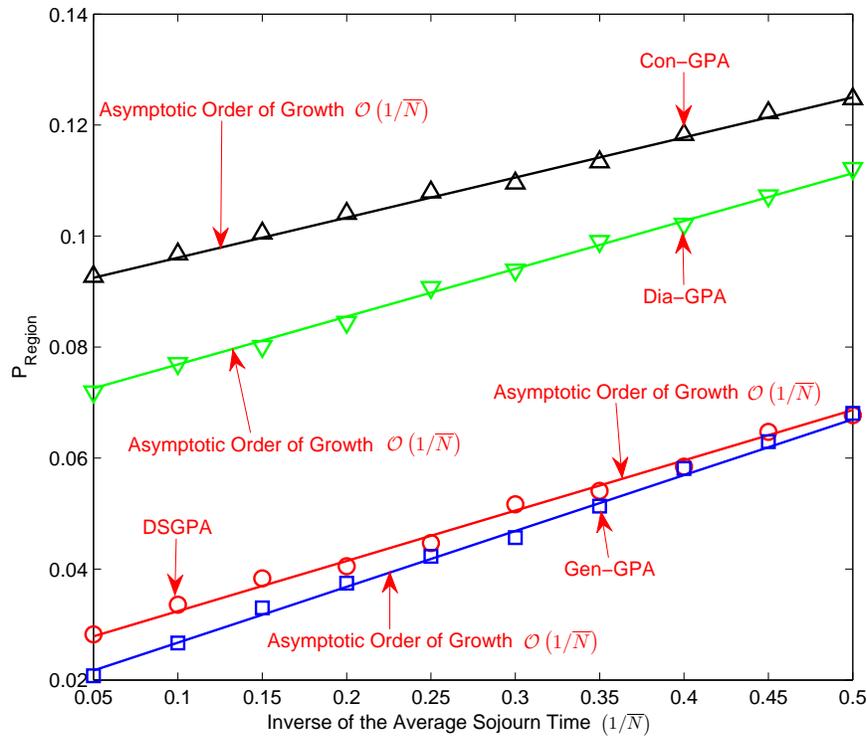}
\caption {\textcolor{black}{Region stability property of the proposed DSGPA and the baseline schemes: Gen-GPA, Dia-GPA and Con-GPA.
The asymptotic order of growth refers to the results stated in {\em Theorem \ref{thm:region_stability}} (see equation \eqref{eqn:region_prob}).
}} \label{fig:sim_2_region}
\end{figure}

\newpage
\begin{figure} [h]
\centering
\includegraphics[scale=0.75]{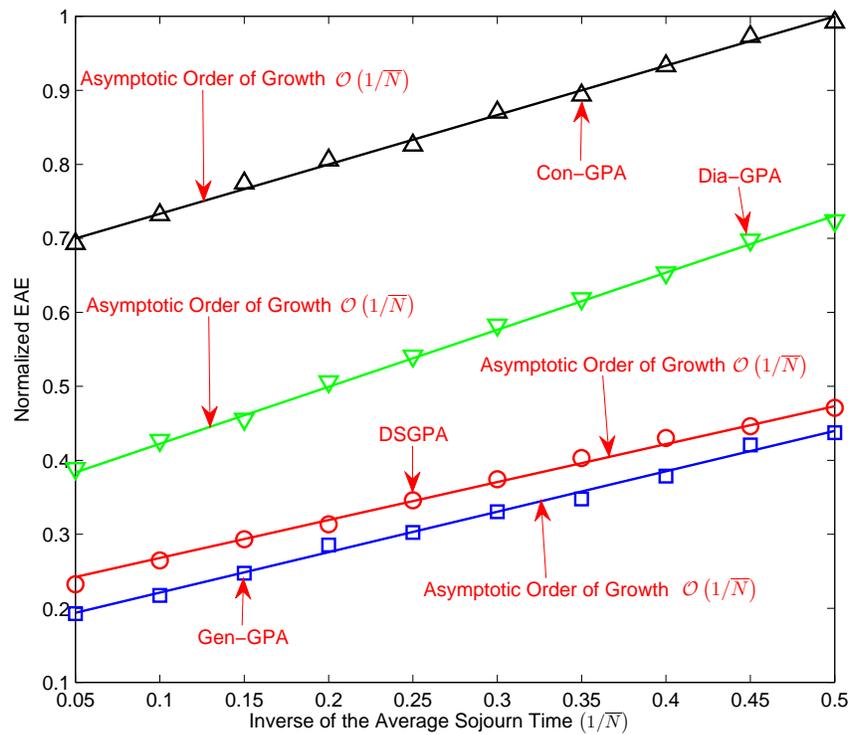}
\caption {\textcolor{black}{Expected-absolute-error (EAE) versus Inverse of the Average Sojourn Time of the proposed DSGPA and the baseline schemes: Gen-GPA, Dia-GPA and Con-GPA.
The EAE is normalized to the maximum EAE of all the schemes.
}} \label{fig:sim_3_EAE}
\end{figure}

\newpage
\begin{figure} [h]
\centering
\includegraphics[scale=0.75]{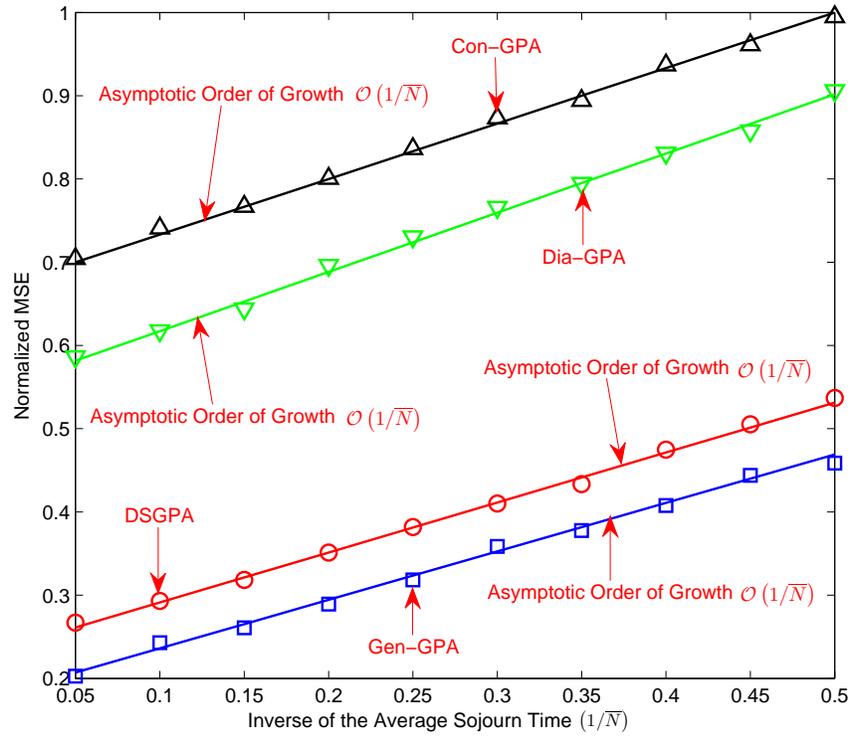}
\caption {\textcolor{black}{Mean-square-error (MSE) versus Inverse of the Average Sojourn Time of the proposed DSGPA and the baseline schemes: Gen-GPA, Dia-GPA and Con-GPA.
The MSE is normalized to the maximum MSE of all the schemes.
}} \label{fig:sim_4_MSE}
\end{figure}

\end{document}